\documentclass[11pt]{article}

\usepackage[active]{srcltx}
\usepackage{hyperref}
\usepackage{amsmath,amssymb}
\usepackage[all]{xy}

 1 

\newtheorem{teor}{Theorem}
\newtheorem{prop}{Proposition}
\newtheorem{corol}{Corollary}
\newtheorem{lem}{Lemma}
\newtheorem{definition}{Definition}

\def\qed{\ifvmode\Realemovelastskip\fi
{\unskip\nobreak\hfil\penalty50\hbox{}\nobreak\hfil \hbox{\vrule
height1.2ex width1.2ex}\parfillskip=0pt \finalhyphendemerits=0
\par\smallskip}}
\def\qedr{\ifvmode\Realemovelastskip\fi
{\unskip\nobreak\hfil\penalty50\hbox{}\nobreak\hfil \hbox{
$\diamond$}\parfillskip=0pt \finalhyphendemerits=0
\par\smallskip}}
\def\ds{\displaystyle}
\newenvironment{proof}{\noindent{\sl Proof:~~~}}{\quad \qed}

\def\beq{\begin{equation}}
\def\eeq{\end{equation}}
\def\bea{\begin{eqnarray}}
\def\eea{\end{eqnarray}}
\def\beann{\begin{eqnarray*}}
\def\eeann{\end{eqnarray*}}
\def\beasn{\begin{sneqnarray}}
\def\eeasn{\end{sneqnarray}}
\def\ben{\begin{enumerate}}
\def\een{\end{enumerate}}
\def\bit{\begin{itemize}}
\def\eit{\end{itemize}}

\def\derpar#1#2{\displaystyle\frac{\partial{#1}}{\partial{#2}}}
\def\derpars#1#2#3{\displaystyle\frac{\partial^2{#1}}{\partial{#2}\partial{#3}}}
\def\restric#1#2{\left.#1\right|_{#2}}

\def\W{{\cal W}}
\def\C{{\cal C}}

\def\vf{{\mathfrak{X}}}
\def\df{{\mit\Omega}}
\def\Lag{{\cal L}}
\def\Leg{{\cal FL}}

\def\d{{\rm d}}

\def\Nat{\mathbb{N}}

\def\Real{\mathbb{R}}
\def\R{\mathbb{R}}

\def\pr{\operatorname{pr}}
\def\Tan{{\rm T}}
\def\Lie{\mathop{\rm L}\nolimits}
\def\inn{\mathop{i}\nolimits}
\def\Cinfty{{\rm C}^\infty}

\def\tabaddress#1{{\small\it\begin{tabular}[t]{c}#1
\\[1.2ex]\end{tabular}}}
\def\qed{\ifvmode\removelastskip\fi
{\unskip\nobreak\hfil\penalty50\hbox{}\nobreak\hfil \hbox{\vrule
height1.2ex width1.2ex}\parfillskip=0pt \finalhyphendemerits=0
\par\smallskip}}

\title{LAGRANGIAN-HAMILTONIAN UNIFIED FORMALISM FOR
AUTONOMOUS HIGHER-ORDER DYNAMICAL SYSTEMS}

\author{
{\sc  Pedro Daniel Prieto-Mart\'\i nez\thanks{{\bf e}-{\it mail}:
   peredaniel@ma4.upc.edu} }\\
   {\sc Narciso Rom\'an-Roy\thanks{{\bf e}-{\it mail}:
   nrr@ma4.upc.edu}}  \\
   \tabaddress{Departamento de Matem\'atica Aplicada IV.
   Edificio C-3, Campus Norte UPC\\
   C/ Jordi Girona 1. 08034 Barcelona. Spain}}

\begin{document}

\maketitle

\pagestyle{myheadings}

\thispagestyle{empty}

\begin{abstract}
The Lagrangian-Hamiltonian unified formalism of R. Skinner and R. Rusk was originally stated
for autonomous dynamical systems in classical mechanics.
It has been generalized for non-autonomous first-order mechanical systems,
as well as for first-order and higher-order field theories.
However, a complete generalization to higher-order mechanical systems has yet to be described.
In this work, after reviewing the natural geometrical setting and the
Lagrangian and Hamiltonian formalisms for higher-order autonomous mechanical systems,
we develop a complete generalization of the Lagrangian-Hamiltonian unified formalism
for these kinds of systems, and we use it to analyze some physical models from this new point of view.
\end{abstract}

 \bigskip
\noindent {\bf Key words}: {\sl Higher-order systems, Lagrangian and 
 Hamiltonian formalisms, Symplectic and presymplectic manifolds.}

\vbox{\raggedleft AMS s.\,c.\,(2000): 70H50, 53C80, 53C15}\null
\markright{\rm P.D. Prieto-Mart\'\i nez, N. Rom\'an-Roy:   \sl Unified formalism for higher-order systems.}

\clearpage

\tableofcontents

\section{Introduction}
\label{section:intro}

In recent decades, a strong development in the intrinsic study of a wide variety of topics in theoretical physics,
control theory and applied mathematics has been done, using methods of differential geometry.
Thus, the intrinsic formulation of Lagrangian and Hamilonian formalisms has been developed
both for autonomous and non-autonomous systems.
This study has been carried out mainly for first-order dynamical systems;
that is, those whose Lagrangian or Hamiltonian functions depend
on the generalized coordinates of position and velocity (or momentum).
From the geometric point of view, this means that the phase space of the system is in most cases
the tangent or cotangent bundle of the smooth manifold representing the configuration space.

However, there are a significant number of relevant systems in which the dynamics have explicit
dependence on accelerations or higher-order derivatives of the generalized coordinates of position.
These systems, usually called {\sl higher-order dynamical systems}, can be modeled geometrically using
higher-order tangent bundles \cite{book:DeLeon_Rodrigues85}.
These models are typical of theoretical physics; for example those describing the interaction of
relativistic particles with spin, string theories from Polyakov and others,
Hilbert's Lagrangian for gravitation or Podolsky's generalization of electromagnetism
(see \cite{art:Batlle_Gomis_Pons_Roman88} and references cited there).
They also appear in a natural way in numerical models arising from the discretization of
first-order dynamical systems that preserve their inherent geometric structures
\cite{art:DeLeon_Martin_Santamaria04}.
There are a lot of works devoted to the development of the formalism of these kinds of theories
and their application to many models in mechanics and field theory (see, for instance, 
\cite{art:Aldaya_Azcarraga78_2},
\cite{art:Aldaya_Azcarraga80},
\cite{art:Banerjee_Mukherjee_Paul10},
\cite{art:Barcelos_Natividade91},
\cite{art:Belvedere_Amaral_Lemos95},
\cite{art:Carinena_Lopez92},
\cite{proc:Garcia_Munoz83},
\cite{art:Krupkova94},
\cite{art:Kuznetsov_Plyushchay94},
\cite{art:Plyushchay91},
\cite{art:Saunders_Crampin90},
\cite{art:Schmidt97}).

Furthermore, a generalization of the Lagrangian and Hamiltonian formalisms exists that compresses them
into a single formalism. This is the so-called {\sl Lagrangian-Hamiltonian unified formalism},
or  {\sl Skinner-Rusk formalism} due to the authors' names of the original paper.
It was originally developed for first-order autonomous mechanical systems \cite{art:Skinner_Rusk83},
and later generalized to non-autonomous dynamical systems
\cite{art:Barbero_Echeverria_Martin_Munoz_Roman08,art:Cortes_Martinez_Cantrijn02},
control systems \cite{art:Barbero_Echeverria_Martin_Munoz_Roman07},
first-order classical field theories \cite{art:DeLeon_Marrero_Martin03,art:Echeverria_Lopez_Marin_Munoz_Roman04}
and, more recently, to higher-order classical field theories
\cite{art:Campos_DeLeon_Martin_Vankerschaver09,art:Vitagliano10}.
Nevertheless, although the geometrization of both higher-order Lagrangian and Hamiltonian formalisms
was already developed for autonomous mechanical systems
\cite{proc:Cantrijn_Crampin_Sarlet86,book:DeLeon_Rodrigues85,art:Gracia_Pons_Roman91,book:Miron10},
a complete generalization of the Skinner-Rusk formalism for higher-order mechanical systems
has yet to be developed.
A first attempt was outlined in \cite{art:Colombo_Martin_Zuccalli10},
with the aim of providing a geometric model for studying optimal control of underactuated systems,
although a deep analysis of the model and its relation with the standard Lagrangian and Hamiltonian formalisms
was not performed.

Thus, the aim of this work is to provide a detailed and complete description of the
Lagrangian-Hamiltonian unified formalism for higher-order autonomous mechanical systems.
Our approach is different from that given in \cite{art:Colombo_Martin_Zuccalli10}
(these differences are commented on Section \ref{section:outlook}).

The paper is organized as follows:

Section \ref{section:structures} consists of a review of the basic definitions and
the geometric structures of higher-order tangent bundles,
some of which are generalizations of the geometric structures of tangent bundles;
namely, the {\sl canonical vector fields}, the {\sl almost-tangent structures} and {\sl semisprays};
whereas others such as the {\sl Tulczyjew derivation} are needed for developing the
Lagrangian and Hamiltonian formalisms of higher-order mechanical systems,
which are also described in this section.
In particular, higher-order regular and singular systems are distinguished.

The main contribution of the work is found in Section \ref{SkinnerRusk},
where the geometric formulation of the Lagrangian-Hamiltonian unified formalism
for higher-order autonomous mechanical systems is described in detail,
including the study of how the Lagrangian and Hamiltonian formalisms are recovered from that formalism.

Finally, in Section \ref{section:examples}, two examples are analyzed in order to show the application of the formalism;
the first is a regular system, the so-called {\sl Pais-Uhlenbeck oscillator},
while the second is a singular one, the {\sl second-order relativistic particle}.

The paper concludes in Section \ref{section:outlook} with a summary of results, discussion and future research.

All the manifolds, the maps and the structures are smooth.
In addition, all the dynamical systems considered are autonomous.
Summation over crossed repeated indices is understood, although on some occasions
the symbol of summation is written explicitly in order to avoid confusion.

\section{Higher-order dynamical systems}
\label{section:structures}

\subsection{Geometric structures of higher-order tangent bundles}

(See \cite{book:DeLeon_Rodrigues85,book:Saunders89,art:Gracia_Pons_Roman91,art:Gracia_Pons_Roman92,phd:Martinez} for details).

\subsubsection{Higher-order tangent bundles}

Let $Q$ be a $n$-dimensional differentiable manifold, and $k\in\Nat$.
The {\sl $k$th-order tangent bundle} of $Q$, denoted by $\Tan^kQ$,  is the $(k+1)n$-dimensional manifold
made of the $k$-jets with source at $0 \in \Real$ and target $Q$; that is, $\Tan^kQ = J_0^k(\Real,Q)$.
It is a submanifold of $J^k(\R,Q)$.

We have the following canonical projections: if $r\leq k$,
$$
\begin{array}{rcclcrccl}
\rho^k_r \colon & \Tan^kQ & \longrightarrow & \Tan^rQ & , &
\beta^k \colon  & \Tan^kQ & \longrightarrow & Q \\
 \\
\ & \tilde{\sigma}^k(0) & \longmapsto & \tilde{\sigma}^r(0) & , &
\ & \tilde{\sigma}^k(0) & \longmapsto & \sigma(0)  \ ,
\end{array}
$$
where $\tilde{\sigma}^k(0)$ denotes a point in $\Tan^kQ$; that is,
the equivalence class of a curve $\sigma \colon I \subset \R \to Q$
by the $k$-jet equivalence relation.
Notice that $\rho^k_0 = \beta^k$, where $\Tan^0Q$ is canonically identified
with $Q$, and $\rho^k_k = {\rm Id}_{\Tan^kQ}$.

If $\left(U,\varphi\right)$ is a local chart in $Q$, with $\varphi = \left(\varphi^A\right)$, $1\leq A \leq n$,
and $\sigma \colon \Real \to Q$ is a curve in $Q$ such that $\sigma(0) \in U$;
by writing $\sigma^A = \varphi^A \circ \sigma$, the $k$-jet $\tilde{\sigma}^k(0)$
is given in $\left(\beta^k\right)^{-1}(U) = \Tan^kU$ by
$\left(q^A,q^A_{1},\ldots,q^A_{k}\right)$, where
$q^A = \sigma^A(0)$ and $\ds q_{i}^A = \frac{d^i\sigma^A}{dt^i}(0)$ ($1\leq i \leq k$).
Usually we write $q_{0}^A$ instead of $q^A$, and so we have the local chart
$\left(\beta^k\right)^{-1}(U)$ in $\Tan^kQ$ with local coordinates $\left(q_{0}^A,q_{1}^A,\ldots,q_{k}^A\right)$.
Local coordinates in $\Tan(\Tan^kQ)$ are denoted by
$\left(q_0^A,q_1^A,\ldots,q_k^A;v_0^A,v_1^A,\ldots,v_k^A\right)$.

Using these coordinates, the local expression of the canonical projections are
$\rho^k_r\left(q_0^A,q_1^A,\ldots,q_k^A\right) = \left(q_0^A,q_1^A,\ldots,q_r^A\right)$,
and then for the tangent maps $\Tan\rho^k_r \colon \Tan(\Tan^kQ) \to \Tan(\Tan^rQ)$,
we have the local expression
$\Tan\rho^k_r\left(q_0^A,q_1^A,\ldots,q_k^A,v_0^A,v_1^A,\ldots,v_k^A\right) = \left(q_0^A,q_1^A,\ldots,q_r^A,v_0^A,v_1^A,\ldots,v_r^A\right)$.

If $\sigma \colon \R \to Q$ is a curve in $Q$,
the {\rm canonical lifting} of $\sigma$ to $\Tan^kQ$
is the curve $\tilde\sigma^k\colon \Real\to\Tan^kQ$ defined as
$\tilde{\sigma}^k(t) = \tilde{\sigma}^k_t(0)$, where $\sigma_t(s) = \sigma(s+t)$,
(that is, the $k$-jet lifting of $\sigma$).
If $k=1$, we will write $\tilde{\sigma}^1 \equiv \tilde{\sigma}$.

Let $V(\rho^k_{r-1})$ be the vertical sub-bundle of $\Tan^kQ$ in $\Tan^{r-1}Q$.
In the above coordinates, for every
 $p \in \Tan^kQ$ and $u\in V_p(\rho^k_{r-1})$, we have that its components are
$u = \left(0,\ldots,0,v_r^A,\ldots,v_k^A\right)$.
Furthermore, if $i_{k-r+1} \colon V(\rho^k_{r-1}) \hookrightarrow \Tan(\Tan^kQ)$
is the canonical embedding, then
$$
i_{k-r+1}\left(q_0^A,\ldots,q_k^A,v_r^A,\ldots,v_k^A\right) = \left(q_0^A,\ldots,q_k^A,0,\ldots,0,v_r^A,\ldots,v_k^A\right) \ .
$$

Consider now the induced bundle of $\tau_{\Tan^{r-1}Q} \colon \Tan(\Tan^{r-1}Q) \to \Tan^{r-1}Q$
by the canonical projection $\rho^k_{r-1}$, denoted by $\Tan^kQ \times_{\Tan^{r-1}Q}\Tan(\Tan^{r-1}Q)$,
which is a vector bundle over $\Tan^kQ$. We have the following commutative diagrams
$$
\xymatrix{
\Tan^kQ \times_{\Tan^{r-1}Q}\Tan(\Tan^{r-1}Q) \ar@{-->}[r] \ar@{-->}[d] 
& \Tan(\Tan^{r-1}Q) \ar[d]^{\tau_{\Tan^{r-1}Q}} \\
\Tan^kQ \ar[r]^{\rho_{r-1}^k} & \Tan^{r-1}Q
}
\quad , \quad
\xymatrix{
\Tan(\Tan^kQ) \ar[rr]^{\Tan\rho_{r-1}^k} \ar[d]^{\tau_{\Tan^kQ}} & \ & \Tan(\Tan^{r-1}Q) \ar[d]^{\tau_{\Tan^{r-1}Q}} \\
\Tan^kQ \ar[rr]^{\rho^k_{r-1}} & \ & \Tan^{r-1}Q \ .
}
$$
Then, there exists a unique bundle morphism
 $s_{k-r+1} \colon \Tan(\Tan^kQ) \to \Tan^kQ \times_{\Tan^{r-1}Q}\Tan(\Tan^{r-1}Q)$
such that the following diagram is commutative:
$$
\xymatrix{
\Tan(\Tan^kQ) \ar@/_/[ddr]_{\tau_{\Tan^kQ}} \ar[dr]^{s_{k-r+1}} \ar@/^/[drr]^{\Tan\rho^k_{r-1}} & \ & \ \\
\ & \Tan^kQ \times_{\Tan^{r-1}Q}\Tan(\Tan^{r-1}Q) \ar@{-->}[r] 
\ar@{-->}[d] & \Tan(\Tan^{r-1}Q) \ar[d]^{\tau_{\Tan^{r-1}Q}} \\
\ & \Tan^kQ \ar[r]^{\rho^k_{r-1}} & \Tan^{r-1}Q  \ . \\
}
$$
It is defined by $s_{k-r+1}(u) = \left(\tau_{\Tan^kQ}(u), \Tan\rho^k_{r-1}(u)\right)$,
for every $u \in \Tan(\Tan^kQ)$.
Its local expression is
$$
s_{k-r+1}\left(q_0^A,\ldots,q_k^A,v_0^A,\ldots,v_k^A\right) = \left(q_0^A,\ldots,q_{r-1}^A,q_{r}^A,\ldots,q_{k}^A,v_0^A,\ldots,v_{r-1}^A\right)\ .
$$
As $s_{k-r+1}$ is a surjective map and ${\rm Im}\,(i_{k-r+1}) = \ker\,(s_{k-r+1})$, we have the exact sequence
$$
0 \longrightarrow V(\rho^k_{r-1})  \stackrel{i_{k-r+1}}{\longrightarrow}
\Tan(\Tan^kQ)  \stackrel{s_{k-r+1}}{\longrightarrow} 
\Tan^kQ \times_{\Tan^{r-1}Q}\Tan(\Tan^{r-1}Q) \longrightarrow 0 \ ,
$$
which is called the {\sl $(k-r+1)$-fundamental exact sequence}.
In local coordinates, it is given by
\begin{align*}
0 \longmapsto & \left(q_0^A,\ldots,q_k^A,v_r^A,\ldots,v_k^A\right) \stackrel{i_{k-r+1}}{\longmapsto}
 \left(q_0^A,\ldots,q_k^A,0,\ldots,0,v_r^A,\ldots,v_k^A\right) \\
& \left(q_0^A,\ldots,q_k^A,v_0^A,\ldots,v_k^A\right) \stackrel{s_{k-r+1}}{\longmapsto}
 \left(q_0^A,\ldots,q_k^A;q_0^A,\ldots,q_{r-1}^A,v_0^A,\ldots,v_{r-1}^A\right) \longmapsto 0 \ .
\end{align*}
Thus, we have $k$ exact sequences
\begin{align*}
1st \colon \ & 0 \longrightarrow V(\rho^k_{k-1}) \stackrel{i_1}{\longrightarrow}
\Tan(\Tan^kQ)  \stackrel{s_1}{\longrightarrow} \Tan^kQ \times_{\Tan^{k-1}Q}\Tan(\Tan^{k-1}Q) \longrightarrow 0 \\
& \vdots \\
rth \colon \ & 0 \longrightarrow V(\rho^k_{k-r}) \stackrel{i_r}{\longrightarrow}
 \Tan(\Tan^kQ) \stackrel{s_r}{\longrightarrow} \Tan^kQ \times_{\Tan^{k-r}Q}\Tan(\Tan^{k-r}Q) \longrightarrow 0 \\
& \vdots \\
kth \colon \ & 0 \longrightarrow V(\beta^k) \stackrel{i_{k}}{\longrightarrow}
 \Tan(\Tan^kQ) \stackrel{s_k}{\longrightarrow} \Tan^kQ \times_{Q}\Tan Q \longrightarrow 0 \ ,
\end{align*}
where $V(\beta^k) \equiv V(\rho^k_0)$ denotes the vertical subbundle of $\Tan^kQ$ on $Q$.
These sequences can be connected by means of the connecting maps
$$
h_{k-r+1} \colon \Tan^kQ \times_{\Tan^{k-r}Q}\Tan(\Tan^{k-r}Q) \longrightarrow V(\rho^k_{r-1})
$$
locally defined as
$$
h_{k-r+1}\left(q_0^A,\ldots,q_{k}^A,v_0^A,\ldots,v_{k-r}^A\right) = 
\left(q_0^A,\ldots,q_k^A, 0,\ldots,0,\frac{r!}{0!}v_0^A,\frac{(r+1)!}{1!}v_1^A,\ldots,\frac{k!}{(k-r)!}v_{k-r}^A\right) \ .
$$
These maps are globally well-defined and are vector bundle isomorphisms over $\Tan^kQ$.
Then we have the following connection between exact sequences:
$$
\xymatrix{
0 \ar[r] & V(\rho^k_{k-r}) \ar[r]^{i_r} & \Tan(\Tan^kQ)
 \ar[r]^-{s_r} & \Tan^kQ \times_{\Tan^{k-r}Q} \Tan(\Tan^{k-r}Q) \ar[dll]^<<{h_{k-r+1}}|(.5){\hole} \ar[r] & 0 \\
0 \ar[r] & V(\rho^k_{r-1}) \ar[r]_{i_{k-r+1}} & \Tan(\Tan^kQ) \ar[r]_-{s_{k-r+1}} 
& \Tan^kQ \times_{\Tan^{r-1}Q} \Tan(\Tan^{r-1}Q) \ar[ull]_<<{h_r} \ar[r] & 0 \ .
}
$$

\subsubsection{Higher-order canonical vector fields. Vertical endomorphisms and almost-tangent structures}
\label{sect:Cap02_LiouvilleVectField}

The {\sl canonical injection} is the map
\begin{equation}
\label{eqn:Cap02_DefCanonicalImmersion}
\begin{array}{rccl}
j_r \colon & \Tan^kQ & \longrightarrow & \Tan(\Tan^{r-1}Q) \\
\ & \tilde{\sigma}^k(0) & \longmapsto & \tilde{\gamma}(0)
\end{array} \quad , \quad (1\leq r\leq k) \ ,
\end{equation}
where
$$
\begin{array}{rccl}
\gamma \colon & \R & \longrightarrow & \Tan^{r-1}Q \\
\ & t & \longmapsto & \gamma(t) = \tilde{\sigma}_t^{r-1}(0) \ .
\end{array}
$$
In local coordinates
\begin{equation}
\label{eqn:Cap02_LocalCoordCanonicalImmersion}
j_r\left(q_0^A,\ldots,q_k^A\right) = \left(q_0^A,\ldots,q_{r-1}^A;q_1^A,q_2^A,\ldots,q_r^A\right) \ .
\end{equation}

Then, the following composition allows us to define a vector field $\Delta_r\in\vf (\Tan^kQ)$,
$$
\xymatrix{
\Tan^kQ \ar[rr]^-{{\rm Id}\times j_{k-r+1}} \ar@/_1.5pc/[rrrrrr]_{\Delta_r} & \ & 
\Tan^kQ \times_{\Tan^{k-r}Q} \Tan(\Tan^{k-r}Q) \ar[rr]^-{h_{k-r+1}} & \ &
 V(\rho^k_{r-1}) \ar[rr]^-{i_{k-r+1}} & \ & \Tan(\Tan^kQ) \ ;
}
$$
that is,
$\Delta_r = i_{k-r+1} \circ h_{k-r+1} \circ \left({\rm Id} \times j_{k-r+1}\right)$.
From the local expressions of $i_{k-r+1}$, $h_{k-r+1}$ and $j_{k-r+1}$ we obtain that
$\Delta_r\left( q_0^A,\ldots,q_k^A \right) = 
\left( q_0^A,\ldots,q_k^A,0,\ldots,0,r!\,q_1^A,(r+1)!\,q_2^A,\ldots,\frac{k!}{(k-r)!}q_{k-r+1}^A \right)$;
or what is equivalent,
$$
\Delta_r = \sum_{i=0}^{k-r} \frac{(r+i)!}{i!} q_{i+1}^A \derpar{}{q_{r+i}^A} = r!\,q_1^A\derpar{}{q_r^A} + 
(r+1)!\,q_2^A\derpar{}{q_{r+1}^A} + \ldots + \frac{k!}{(k-r)!}\,q_{k-r+1}^A\derpar{}{q_k^A} \ .
$$
In particular
$$
\Delta_1 = \sum_{i=1}^{k} i q_i^A \derpar{}{q_{i}^A} = \sum_{i=0}^{k-1} (i+1) q_{i+1}^A \derpar{}{q_{i+1}^A} = 
q_1^A\derpar{}{q_1^A} + 2q_2^A\derpar{}{q_2^A} + \ldots + kq_{k}^A\derpar{}{q_k^A} \ .
$$

\begin{definition}
The vector field $\Delta_r$ is the {\rm $r$th-canonical vector field} in $\Tan^kQ$. In particular,
$\Delta_1$ is called the {\rm Liouville vector field} in $\Tan^kQ$.
\end{definition}


Remember that, if $N$ is a $(k+1)n$-dimensional manifold,
an {\sl almost-tangent structure of order $k$} in $N$ is an endomorphism $J$ in $\Tan N$ such that
$J^{k+1} = 0$ and ${\rm rank}\,J = kn$. Then,
$\Tan^kQ$ is endowed with a canonical almost-tangent structure. In fact:

\begin{definition}
For $1 \leq r \leq k$, let $i_{k-r+1}$, $h_{k-r+1}$, $s_r$ be the morphisms of the fundamental exact sequences
introduced above. The map
$$
J_r = i_{k-r+1} \circ h_{k-r+1} \circ s_r \colon \Tan(\Tan^kQ) \longrightarrow \Tan(\Tan^kQ)
$$
defined by the composition
$$
\xymatrix{
\Tan(\Tan^kQ) \ar[rr]^-{s_r} \ar@/_1.5pc/[rrrrrr]_{J_r} & \ & 
\Tan^kQ \times_{\Tan^{k-r}Q} \Tan(\Tan^{k-r}Q) \ar[rr]^-{h_{k-r+1}} & \ & 
V(\rho^k_{r-1}) \ar[rr]^-{i_{k-r+1}} & \ & \Tan(\Tan^kQ)
}
$$
is called the {\rm $r$th-vertical endomorphism} of $\Tan(\Tan^kQ)$.
\end{definition}

From the local expressions of $s_r$, $h_{k-r+1}$, $i_{k-r+1}$  we obtain that
$$
J_r\left(q_0^A,\ldots,q_k^A,v_0^A,\ldots,v_k^A\right) =
 \left(q_0^A,\ldots,q_k^A,0,\ldots,0,r!\,v_0^A,(r+1)!\,v_1^A,\ldots,\frac{k!}{(k-r)!}\,v_{k-r}^A\right) \ ;
$$
that is,
$\ds J_r = \sum_{i=0}^{k-r} \frac{(r+i)!}{i!} \, dq_i^A \otimes \derpar{}{q_{r+i}^A}$.
In particular,
$\ds J_1 = \sum_{i=0}^{k-1} (i+1) dq_i^A \otimes \derpar{}{q_{i+1}^A}$.

The $r$th-vertical endomorphism $J_r$ has constant rank $(k-r+1)n$ and satisfies that
$$
\left(J_r\right)^s = \begin{cases} 0 & \mbox{\rm if } rs \geqslant k+1 \\ J_{rs} & \mbox{\rm if } rs < k\end{cases} \ .
$$

As a consequence, the $1$st-vertical endomorphism $J_1$ defines an almost-tangent structure
of order $k$ in $\Tan^kQ$, which is called the
{\sl canonical almost-tangent structure} of $\Tan^kQ$.
Then, any other vertical endomorphism $J_r$ is obtained by composing $J_1$ with itself $r$ times.
Furthermore, we have the following relation:
$$
J_r\circ \Delta_s = \begin{cases} 0, & \mbox{\rm if } r+s\geqslant k+1 \\
 \Delta_{r+s}, & \mbox{\rm if } r+s < k+1 \end{cases}
 $$
As a consequence, starting from the Liouville vector field and the vertical endomorphisms,
we can recover all the canonical vector fields. However, as all the vertical endomorphisms are obtained from $J_1$,
we conclude that all the canonical structures in $\Tan^kQ$ are obtained from the Liouville vector field
and the canonical almost-tangent structure.

Consider now the dual maps $J_r^*$ of $J_r$, $1 \leqslant r \leqslant k$;
that is, the endomorphisms in $\Tan^*(\Tan^kQ)$, 
and their natural extensions to the exterior algebra $\bigwedge(\Tan^*(\Tan^kQ))$ 
(also denoted by $J_r^*$). Their action on the set of differential forms is given by
\begin{equation*}
J_r^*\omega(X_1,\ldots,X_p) = \omega(J_r(X_1),\ldots,J_r(X_p)) \ ,
\end{equation*}
for $\omega\in\df^p(\Tan^kQ)$ and $X_1,\ldots,X_p \in \vf(\Tan^kQ)$,
and for every $f\in\Cinfty(\Tan^kQ)$ we write $J_r^*(f) = f$.
The endomorphism
$J_r^* \colon \df(\Tan^kQ) \to \df(\Tan^kQ)$, $1\leq r \leq k$, 
is called the {\sl $r$th-vertical operator}, and it is locally given by
\begin{align*}
&J_r^*(f) = f \quad , \quad \mbox{\rm for every } \ f \in \Cinfty(\Tan^kQ) \\
&J_r^*(\d q_i^A) = \begin{cases} 0, & \mbox{if }i < r \\ \frac{i!}{(i-r)!}\, \d q_{i-r}^A, & \mbox{if }i\geq r
 \end{cases} \ .
\end{align*}

\subsubsection{Vertical derivations and differentials. Tulczyjew's derivation}

The {\sl inner contraction} of the vertical endomorphisms $J_r$ with any differential $p$-form
$\omega\in\df^p(\Tan^kQ)$ is the $p$-form $\inn(J_r)\omega$ defined as follows: for every
$X_1,\ldots,X_p\in\vf(\Tan^kQ)$
$$
\inn(J_r)\omega(X_1,\ldots,X_p) = \sum_{i=1}^{p} \omega(X_1,\ldots,J_r(X_i),\ldots,X_p)\ ,
$$
and taking $\inn(J_r)f = 0$, for every $f\in\Cinfty(\Tan^kQ)$, we can state:

\begin{definition}
The map
$$
\begin{array}{rcl}
 \df(\Tan^kQ) & \longrightarrow & \df(\Tan^kQ) 
\\ \omega & \longmapsto & \inn (J_r)\omega
\end{array}
$$
is a derivation of degree $0$ in $\df(\Tan^kQ)$, which is called the
{\rm $r$th-vertical derivation} in $\df(\Tan^kQ)$.
\end{definition}

Its local expression is
$$
\inn(J_r)(\d q_i^A) = \begin{cases} 0, & \mbox{\rm if }i<r \\ \frac{i!}{(i-r)!}\,\d q_{i-r}^A, & \mbox{\rm if }i\geq r \end{cases} \ .
$$

\begin{definition}
The operator $d_{J_r} = [\inn(J_r),\d]$ is a skew-derivation of degree $1$,
which is called the {\rm $r$th-vertical differential}.
\end{definition}

Its local expression is given by
$$
\begin{array}{l} \displaystyle d_{J_r}(f) = \sum_{i=r}^k \frac{i!}{(i-r)!} \derpar{f}{q_i^A} \d q_{i-r}^A
\quad , \quad \mbox{\rm for every $f\in\Cinfty(\Tan^kQ)$} \\
d_{J_r}(\d q^i) = 0 \end{array} \ .
$$

For $1 \leq r,s \leq k$, we have that
 $ d_{J_r}\d = -\d d_{J_r}$.
 

In the set $\oplus_{k\geqslant 0}\df(\Tan^kQ)$,
we can define the following equivalence relation:
for $\omega \in \df(\Tan^kQ)$ and $\lambda \in \df(\Tan^{k'}Q)$,
$$
\omega \sim \lambda \Longleftrightarrow \begin{cases} \omega = (\rho^k_{k'})^*(\lambda), & \mbox{if }k'\leqslant k \\ 
\lambda = (\rho^{k'}_k)^*(\omega), & \mbox{if }k' \geqslant k \end{cases} \ .
$$
Then we consider the quotient set
$\ds \mit\Omega = \bigoplus_{k\geqslant0}\df(\Tan^kQ)/ \sim$,
which is a commutative graded algebra.
In this set we can define the so-called
{\sl Tulczyjew's derivation} \cite{art:Tulczyjew75_1,book:DeLeon_Rodrigues85}, denoted by $d_T$, as follows:
for every $f \in \Cinfty(\Tan^kQ)$ we construct the function $d_Tf \in \Cinfty(\Tan^{k+1}Q)$
given by
$$
(d_Tf)(\tilde{\sigma}^{k+1}(0)) = (d_{\tilde{\sigma}^k(0)}f)(j_{k+1}(\tilde{\sigma}^{k+1}(0)))
$$
where $j_{k+1} \colon \Tan^{k+1}Q \to \Tan(\Tan^kQ)$ is the canonical injection introduced in
the Section \ref{sect:Cap02_LiouvilleVectField}.
From the coordinate expression for $j_{k+1}$, we obtain that
$$
d_Tf\left(q_0^A,\ldots,q_{k+1}^A\right) =
  \sum_{i=0}^{k}q_{i+1}^A \derpar{f}{q_i^A}(q_0^A,\ldots,q_{k}^A) \ .
$$
This map $d_T$ extends to a derivation of degree $0$ in $\mit\Omega$ and,
as $d_T\d = \d d_T$, it is determined by its action on functions
and by the property
$d_T(\d q_i^A) = \d q_{i+1}^A$.

Furthermore,  the maps $\inn(J_s)$, $d_{J_s}$, $\inn(\Delta_s)$ and $\Lie(\Delta_s)$ extend to $\mit\Omega$
in a natural way.

\subsubsection{Higher-order semisprays}

\begin{definition}
A vector field $X\in\vf(\Tan^kQ)$ is a {\rm semispray of type $r$}, $1 \leq r \leq k$,
if for every integral curve $\sigma$ of $X$, we have that, if
 $\gamma=\beta^k \circ \sigma$, then
 $\tilde\gamma^{k-r+1} = \rho^k_{k-r+1}\circ\sigma$
(where $\tilde\gamma^{k-r+1}$
is the canonical lifting of $\gamma$  to $\Tan^{k-r+1}Q$).
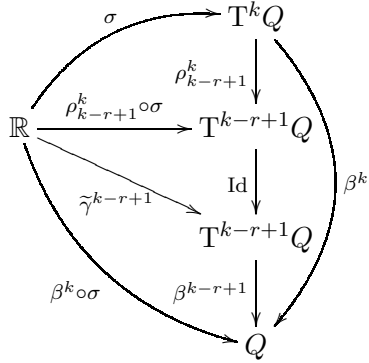
$$
\xymatrix{
\ & \ & \Tan^kQ \ar[d]_{\rho^k_{k-r+1}} \ar@/^2.5pc/[ddd]^{\beta^k} \\
\R \ar@/^1.5pc/[urr]^{\sigma} \ar@/_1.5pc/[ddrr]_{\beta^k\circ\sigma}
 \ar[rr]^-{\rho^k_{k-r+1}\circ\sigma} \ar[drr]_{\widetilde{\gamma}^{k-r+1}} 
 & \ & \Tan^{k-r+1}Q \ar[d]_{{\rm Id}} \\
\ & \ & \Tan^{k-r+1}Q \ar[d]_{\beta^{k-r+1}} \\
\ & \ & Q
}
$$
In particular, $X\in\vf(\Tan^kQ)$ is a {\rm semispray of type $1$}
if for every integral curve $\sigma$ of $X$, we have that
 $\gamma=\beta^k \circ \sigma$, then
 $\tilde\gamma^k=\sigma$.
\end{definition}

The local expression of a semispray of type $r$ is
$$
X = q_1^A\derpar{}{q_0^A} + q_2^A\derpar{}{q_1^A} + \ldots + q_{k-r+1}^A\derpar{}{q_{k-r}^A} + 
X_{k-r+1}^A\derpar{}{q_{k-r+1}^A} + \ldots + X_k^A\derpar{}{q_k^A}
$$

\begin{prop}
The following assertions are equivalent:
\begin{enumerate}
\item
A vector field $X\in\vf(\Tan^kQ)$ is a semispray of type $r$.
\item
 $\Tan\rho^k_{k-r} \circ X = j_{k-r+1}$;
 that is, the following diagram commutes
$$
\xymatrix{
\Tan(\Tan^kQ) \ar[drr]^{\Tan\rho^k_{k-r}} \\
\Tan^kQ \ar[u]^X \ar[rr]^-{j_{k-r+1}} & \ & \Tan(\Tan^{k-r}Q) \ .
}
$$
\item
$J_r(X) = \Delta_r$.
\end{enumerate}
\end{prop}

Obviously, every semispray of type $r$ is a semispray of type $s$, for $s\geq r$.

If $X\in\vf(\Tan^kQ)$ is a semispray of type $r$, a curve $\sigma$ in $Q$
is said to be a {\sl path} or {\sl solution} of $X$ if
 $\tilde{\sigma}^k$ is an integral curve of $X$; that is,
$\widetilde{\tilde{\sigma}^k} = X \circ \tilde{\sigma}^k$,
where $\widetilde{\tilde{\sigma}^k}$ denotes the canonical lifting of $\tilde{\sigma}^k$
from $\Tan^kQ$ to $\Tan(\Tan^kQ)$.
Then, in coordinates, $\sigma$ verifies the following system of differential equations of order $k+1$:
\begin{align*}
\frac{d^{k-r+2}\sigma^A}{dt^{k-r+2}} &= X_{k-r+1}^A\left(\sigma,\frac{d\sigma}{dt},\ldots,\frac{d^k\sigma}{dt^k}\right)\\
& \; \vdots \\
\frac{d^{k+1}\sigma^A}{dt^{k+1}} &= X_k^A\left(\sigma,\frac{d\sigma}{dt},\ldots,\frac{d^k\sigma}{dt^k}\right)
\end{align*}

Observe that, taking $k=1$, then $r=1$ and $\rho^1_{1-1+1} = {\rm Id}_{\Tan Q}$,
we recover the definition of the holonomic vector field ({\sc sode} in $\Tan Q$).
So, semisprays of type $1$ in $\Tan^kQ$ are the analogue to the holonomic vector fields in $\Tan Q$; that is,
they are the vector fields whose integral curves are the canonical liftings
to $\Tan^kQ$ of curves on the basis $Q$.
Their local expressions are
$$
X = q_1^A\derpar{}{q_0^A} + q_2^A\derpar{}{q_1^A} + \ldots + q_k^A\derpar{}{q_{k-1}^A} + X_k^A\derpar{}{q_k^A}\ .
$$

\subsection{Lagrangian formalism}


Let $Q$ be a $n$-dimensional differentiable manifold and $\Lag \in \Cinfty(\Tan^kQ)$.
We say that $\Lag$ is a Lagrangian function of order $k$. 

\begin{definition}
The {\rm Lagrangian $1$-form} $\theta_\Lag \in \df^1(\Tan^{2k-1}Q)$,
associated to $\Lag$ is defined as
$$
\theta_\Lag = \sum_{r=1}^k (-1)^{r-1} \frac{1}{r!} d_T^{r-1} d_{J_r}\Lag  \ .
$$
Then, the {\rm Lagrangian $2$-form}, $\omega_\Lag \in \df^2(\Tan^{2k-1}Q)$,
associated to $\Lag$ is
$$
\omega_\Lag = -\d \theta_\Lag = \sum_{r=1}^k (-1)^r \frac{1}{r!} d_T^{r-1}\d d_{J_r}\Lag \ .
$$
\end{definition}

Observe that the Lagrangian $1$-form is a semibasic form of type $k$ in  $\Tan^{2k-1}Q$ .

We assume that 
 $\omega_\Lag$ has constant rank (we refer to this fact by saying that $\Lag$ is a
 {\sl geometrically admissible Lagrangian}).
 
\begin{definition}
The {\rm Lagrangian energy}, $E_\Lag \in \Cinfty(\Tan^{2k-1}Q)$, associated to $\Lag$
is defined as
$$
E_\Lag = \left(\sum_{r=1}^k (-1)^{r-1} \frac{1}{r!} d_T^{r-1}(\Delta_r(\Lag))\right) - (\rho_k^{2k-1})^*\Lag
$$
\end{definition}

It is usual to write $\Lag$ instead of $(\rho_k^{2k-1})^*\Lag$,
and we will do this in the sequel.

The coordinate expressions of these elements are
\bea
\label{eqn:Cap03_LocalCoordLag1Form}
\theta_\Lag &=& \sum_{r=1}^k \sum_{i=0}^{k-r}(-1)^i d_T^i\left(\derpar{L}{q_{r+i}^A}\right) \d q_{r-1}^A \\
\omega_\Lag &=&
\sum_{r=1}^k \sum_{i=0}^{k-r}(-1)^{i+1} d_T^i\,\d\left(\derpar{\Lag}{q_{r+i}^A}\right) \wedge \d q_{r-1}^A \nonumber \\
\label{eqn:Cap03_LocalCoordLagEnergy}
E_\Lag &=& \sum_{r=1}^{k} q_{r}^A \sum_{i=0}^{k-r} (-1)^i d_T^i\left( \derpar{L}{q_{r+i}^A} \right)- \Lag \ .
\eea

\begin{definition}
A Lagrangian function $\Lag \in \Cinfty(\Tan^kQ)$ is said to be {\rm regular}
if $\omega_\Lag$ is a symplectic form. Otherwise $\Lag$ is a {\rm singular} Lagrangian.
\end{definition}

To say that $\Lag$ is a regular Lagrangian is locally equivalent to saying that
the Hessian matrix
$\ds \left(\frac{\partial^2\Lag}{\partial q_k^B \partial q_k^A}\right)$ is regular at every point of $\Tan^kQ$.

\begin{definition}
A {\rm Lagrangian system of order $k$} is a couple $(\Tan^{2k-1}Q,\Lag)$, where $Q$
represents the configuration space and $\Lag \in \Cinfty(\Tan^kQ)$
is the Lagrangian function.
It is said to be a {\rm regular} (resp. {\rm singular}) Lagrangian system if the Lagrangian function $\Lag$ is
regular (resp. singular).
\end{definition}

Thus, in the Lagrangian formalism, $\Tan^{2k-1}Q$ represents the phase space of the system.
The dynamical trajectories of the system are the integral curves of any
vector field $X_\Lag \in \vf(\Tan^{2k-1}Q)$ satisfying that:
\begin{enumerate}
\item
It is a solution to the equation
\begin{equation}\label{eqn:Cap03_IntrinsicLagEq}
\inn(X_\Lag)\omega_\Lag = \d E_\Lag
\end{equation}
\item 
It is a semispray of type $1$ in $\Tan^{2k-1}Q$.
\end{enumerate}
Equation (\ref{eqn:Cap03_IntrinsicLagEq}) is the
{\sl higher-order Lagrangian equation}, and a vector field
$X_\Lag$ solution to (\ref{eqn:Cap03_IntrinsicLagEq})  (if it exists)
is called a {\sl Lagrangian vector field of order $k$}.
If, in addition, $X_\Lag$ satisfies condition 2, then it is called an
{\sl Euler-Lagrange vector field of order $k$}, and its integral curves on the base are
solutions to the {\sl higher-order Euler-Lagrange equations}.

In natural coordinates of $\Tan^{2k-1}Q$, if
$$
X_\Lag = \sum_{i=0}^{2k-1} f_i^A \derpar{}{q_i^A} = 
f_0^A\derpar{}{q_0^A} + f_1^A \derpar{}{q_1^A} + \ldots + f_{2k-1}^A\derpar{}{q_{2k-1}^A} \	 ,
$$
as
$$
\d E_\Lag = \sum_{r=1}^k \sum_{i=0}^{k-r}(-1)^i d_T^i \left(\derpar{\Lag}{q_{r+i}^A} \right)\d q_r^A +
 \sum_{r=1}^k q_r^A \sum_{i=0}^{k-r} (-1)^i\sum_{j=0}^{k} d_T^i\left(\frac{\partial^2\Lag}{\partial q_j^B\partial q_{r+i}^A}
 \d q_j^B \right) - \sum_{r=0}^{k} \derpar{\Lag}{q_r^A} \d q_r^A \ ,
$$
from (\ref{eqn:Cap03_IntrinsicLagEq}) we obtain
\begin{equation}
\label{eqn:Cap03_LocalCoordLagEq}
\begin{array}{l}
\displaystyle \left(f_0^B-q_1^B\right) \frac{\partial^2\Lag}{\partial q_k^B\partial q_k^A} = 0 \\[10pt]
\displaystyle \left(f_{1}^B - q_{2}^B\right)\frac{\partial^2\Lag}{\partial q_k^B\partial q_k^A} - \left(f_0^B-q_{1}^B \right)(\cdots\cdots) = 0 \\
\qquad \qquad \qquad \qquad \vdots \\
\displaystyle \left(f_{2k-2}^B - q_{2k-1}^B\right)\frac{\partial^2\Lag}{\partial q_k^B\partial q_k^A} -
 \sum_{i=0}^{2k-3} \left(f_{i}^B-q_{i+1}^B \right) (\cdots\cdots) = 0 \\
\displaystyle (-1)^k\left(f_{2k-1}^B - d_T\left(q_{2k-1}^B\right)\right) \frac{\partial^2\Lag}{\partial q_k^B\partial q_k^A} +
 \sum_{i=0}^{k} (-1)^id_T^i\left( \derpar{L}{q_i^A} \right) - \sum_{i=0}^{2k-2} \left(f_{i}^B-q_{i+1}^B \right) (\cdots\cdots) = 0 \,
\end{array}
\end{equation}
where the terms in brackets $(\cdots\cdots)$ contain relations involving partial derivatives
of the Lagrangian and applications of $d_T$, which for simplicity are not written.
These are the local expressions of the Lagrangian equations for $X_\Lag$.

Now, if $\sigma \colon \Real \to \Tan^{2k-1}Q$ is an integral curve of $X_\Lag$,
from (\ref{eqn:Cap03_IntrinsicLagEq}) we obtain that $\sigma$ must satisfy the
{\sl Euler-Lagrange equation}
\begin{equation}
\label{eqn:Cap03_IntrinsicLagEqCI}
\inn(\tilde{\sigma})(\omega_\Lag \circ \sigma) = \d E_\Lag \circ \sigma \ ,
\end{equation}
where $\tilde{\sigma}$ denotes the canonical lifting of $\sigma$ to $\Tan(\Tan^{2k-1}Q)$;
and as $X_\Lag$ is a semispray of type $1$, we have that $\sigma$ is the
canonical lifting of a curve $\gamma \colon \Real \to Q$ to $\Tan^{2k-1}Q$;
that is, $\sigma = \tilde{\gamma}^{2k-1}$.

Now, if $\Lag \in \Cinfty(\Tan^kQ)$ is a regular Lagrangian,
then $\omega_\Lag$ is a symplectic form in $\Tan^{2k-1}Q$, and
as a consequence we have that:

\begin{teor}
\label{prop:Cap03_LagVectFieldRegLag}
Let $(\Tan^{2k-1}Q,\Lag)$ be a regular Lagrangian system of order $k$.
\ben
\item
There exists a unique $X_\Lag \in \vf(\Tan^{2k-1}Q)$
which is a solution to the Lagrangian equation (\ref{eqn:Cap03_IntrinsicLagEq})
and is a semispray of type $1$ in $\Tan^{2k-1}Q$.
\item
If $\gamma \colon \Real \to Q$ is an integral curve of $X_\Lag$ then
$\sigma=\tilde{\gamma}^{2k-1}$ is a solution to the 
 {\rm Euler-Lagrange equations}:
 \begin{equation}
\label{eqn:Cap03_EulerLagrangeEquations}
\derpar{\Lag}{q^0} \circ \tilde{\gamma}^{2k-1} - \frac{d}{dt}\derpar{\Lag}{q^1} \circ \tilde{\gamma}^{2k-1} +
 \frac{d^2}{dt^2}\derpar{\Lag}{q^2} \circ\tilde{\gamma}^{2k-1}+ \ldots + 
 (-1)^k\frac{d^k}{dt^k}\derpar{\Lag}{q^k} \circ \tilde{\gamma}^{2k-1} = 0 \ .
\end{equation}
\een
\end{teor}
%

If $\Lag \in \Cinfty(\Tan^kQ)$ is a singular Lagrangian, then $\omega_\Lag$ is a presymplectic form,
so the existence and uniqueness of solutions to the Lagrangian equation (\ref{eqn:Cap03_IntrinsicLagEq}) 
is not assured, except in special cases
(for instance, when $\omega_\Lag$ is a {\sl presymplectic horizontal structure} \cite{book:DeLeon_Rodrigues85}).
In general, in the most favourable cases, equation (\ref{eqn:Cap03_IntrinsicLagEq})
has solutions $X_\Lag \in \vf(\Tan^{2k-1}Q)$ in some submanifold $S_f\hookrightarrow\Tan^{2k-1}Q$,
for which these vector fields solution are tangent.
This submanifold is obtained by applying the well-known constraint algorithms
(see, for instance, \cite{art:Gotay_Nester_Hinds78,art:Gotay_Nester79,art:Munoz_Roman92}).
Nevertheless, these vector fields solution are not necessarily semisprays of type $1$ on $S_f$,
but only on the points of another submanifold $M_f\hookrightarrow S_f\hookrightarrow \Tan^{2k-1}Q$
(see \cite{art:Gotay_Nester79,art:Munoz_Roman92}).
On the points of this last submanifold, the integral curves of
$X_\Lag \in \vf(\Tan^{2k-1}Q)$ are solutions to the Euler-Lagrange equations
(\ref{eqn:Cap03_EulerLagrangeEquations}).

A detailed study of higher-order singular Lagrangian systems
can be found in \cite{art:Gracia_Pons_Roman91,art:Gracia_Pons_Roman92}.

\subsection{Hamiltonian formalism}
\label{subsection:hamform}

\begin{definition}
Let $(\Tan^{2k-1}Q,\Lag)$ be a Lagrangian system.
The {\rm Legendre-Ostrogradsky map} (or {\rm generalized Legendre map\/}) associated to $\Lag$
is the map $\Leg \colon \Tan^{2k-1}Q \to \Tan^*(\Tan^{k-1}Q)$ defined as follows:
 for every $u \in \Tan(\Tan^{2k-1}Q)$,
 $$
\theta_\Lag(u) = \left\langle \Tan \rho_{k-1}^{2k-1}(u), \Leg(\tau_{\Tan^{2k-1}Q}(u)) \right\rangle
$$
\end{definition}

This map verifies that $\pi_{\Tan^{k-1}Q} \circ \Leg = \rho^{2k-1}_{k-1}$,
where $\pi_{\Tan^{k-1}Q} \colon \Tan^*(\Tan^{k-1}Q) \to \Tan^{k-1}Q$ is the natural projection.
Furthermore, if $\theta_{k-1}\in\df^1(\Tan^*(\Tan^{k-1}Q)$ and
$\omega_{k-1}=-\d\theta_{k-1}\in\df^2(\Tan^*(\Tan^{k-1}Q))$
are the canonical $1$ and $2$ forms of the cotangent bundle $\Tan^*(\Tan^{k-1}Q)$, we have that
$$
\Leg^*\theta_{k-1} = \theta_\Lag
\quad , \quad
\Leg^*\omega_{k-1} =  \omega_\Lag \ .
$$

Given a local natural chart in $\Tan^{2k-1}Q$, we can define the following local functions
$$
\hat p^{r-1}_A = \sum_{i=0}^{k-r}(-1)^i d_T^i\left(\derpar{L}{q_{r+i}^A}\right) \ .
$$
Observe that
\begin{align*}
\hat p^{r-1}_A - \derpar{\Lag}{q_r^A} &= \sum_{i=0}^{k-r}(-1)^i d_T^i\left(\derpar{\Lag}{q_{r+i}^A}\right) - \derpar{\Lag}{q_r^A}
= \sum_{i=1}^{k-r} (-1)^i d_T^i\left(\derpar{\Lag}{q_{r+i}^A}\right) \\
&= \sum_{i=0}^{k-r-1} (-1)^{i+1} d_T^{i+1} \left( \derpar{\Lag}{q_{r+i+1}^A}\right)
= -d_T\left(\sum_{i=0}^{k-(r+1)} (-1)^{i}d_T^i \left( \derpar{\Lag}{q_{(r+1)+i}^A}\right) \right) = -d_T(\hat p^r_A)
\end{align*}
and hence
\begin{equation}
\label{eqn:Cap04_MomentumCoordRelation}
\hat p^{r-1}_A = \derpar{\Lag}{q_r^A} - d_T(\hat p^r_A) \quad , \quad 1 \leq r \leq k-1 \ .
\end{equation}
Thus, bearing in mind the local expression (\ref{eqn:Cap03_LocalCoordLag1Form})
of the form $\theta_\Lag$, we can write
$\theta_\Lag = \sum_{r=1}^k \hat p^{r-1}_A\d q_{r-1}^A$, and
we obtain that the expression in natural coordinates of the map $\Leg$ is
$$
\Leg\left(q_0^A,q_1^A,\ldots,q_{2k-1}^A\right) = \left(q_0^A,q_1^A,\ldots,q_{k-1}^A,p^0_A,p^1_A,\ldots,p^{k-1}_A\right) \ , \
\mbox{\rm with $p^i_A\circ\Leg=\hat p^i_A$} \ .
$$

$\Lag$ is a regular Lagrangian if, and only if, 
$\Leg \colon \Tan^{2k-1}Q \to \Tan^*(\Tan^{k-1}Q)$ is a local diffeomorphism.

As a consequence of this, we have that, 
if $\Lag$ is a regular Lagrangian, then the set
$(q_i^A,\hat p^i_A)$, $0\leq i\leq k-1$, is a set of local coordinates in $\Tan^{2k-1}Q$,
and $(\hat p^i_A)$ are called the {\sl Jacobi-Ostrogradsky momentum coordinates}.

Observe that the relation (\ref{eqn:Cap04_MomentumCoordRelation})
means that we can recover all the Jacobi-Ostrogadsky momentum coordinates from the set $(\hat p^{k-1}_A)$.

\begin{definition}
$\Lag \in \Cinfty(\Tan^{k}Q)$ is said to be a {\rm hyperregular Lagrangian} of order $k$
if  $\Leg$ is a global diffeomorphism.
Then, $(\Tan^{2k-1}Q,\Lag)$ is a {\rm hyperregular Lagrangian system} of order $k$.
\end{definition}

As $\pi_{\Tan^{k-1}Q} \circ \Leg = \rho^{2k-1}_{k-1}$, 
this condition is equivalent to demanding that the restriction of
$\rho^{2k-1}_{k-1} \colon \Tan^{2k-1}Q \to \Tan^{k-1}Q$ to every fibre be one-to-one.

In order to explain the construction of the canonical Hamiltonian formalism of a Lagrangian higher-order system,
we first consider the case of hyperregular systems (the regular case is the same, but restricting
on the suitable open submanifolds where $\Leg$ is a local diffeomorphism).

So, $(\Tan^{2k-1}Q,\Lag)$ being a hyperregular Lagrangian system,
since $\Leg$ is a diffeomorphism, there exists a unique function
$h \in \Cinfty(\Tan^*(\Tan^{k-1}Q))$ such that $\Leg^*h = E_\Lag$,
which is called the {\sl Hamiltonian function} associated to this system.
Then the triad $(\Tan^*(\Tan^{k-1}Q),\omega_{k-1},h)$ is called the
{\sl canonical Hamiltonian system} associated to the hyperregular Lagrangian
system $(\Tan^{2k-1}Q,\Lag)$.
Thus, in the Hamiltonian formalism, $\Tan^*(\Tan^{k-1}Q)$ represents the phase space of the system.

The dynamical trajectories of the system are the integral curves of 
a vector field $X_h \in \vf(\Tan^*(\Tan^{k-1}Q))$ which is a solution to the {\sl Hamilton equation}
\begin{equation}
\label{eqn:Cap04_IntrinsicHamEq}
\inn(X_h)\omega_{k-1} = \d h \ .
\end{equation}
As $\omega_{k-1}$ is symplectic, there is a unique vector field $X_h$ solution to this equation,
and it is called the {\sl Hamiltonian vector field}.

In natural coordinates of $\Tan^*(\Tan^{k-1}Q)$, $(q_i^A,p^i_A)$ (with $0\leq i\leq k-1$; $1\leq A\leq n$), 
taking
$\ds X_h = f_i^A \derpar{}{q_i^A} + g^i_A \derpar{}{p^i_A}$,
as $\ds \d h = \derpar{h}{q_i^A}\d q_i^A + \derpar{h}{p^i_A}\d p^i_A$, and
$\omega_{k-1} =\d q_i^A \wedge\d p^i_A$, from (\ref{eqn:Cap04_IntrinsicHamEq}) we obtain that
$$
f_i^A  = \derpar{h}{p^i_A} \quad , \quad g^i_A = -\derpar{h}{q_i^A} \ .
$$

Now, if $\sigma \colon \Real \to\Tan^*(\Tan^{k-1}Q)$ is an integral curve of $X_h$,
we have that $\sigma$ must satisfy the {\sl Hamiltonian equation}
$$
\inn(\tilde{\sigma})(\omega_{k-1} \circ \sigma) = \d h \circ \sigma \ ,
$$
and, if $\sigma(t)=(q_i^A(t),p^i_A(t))$ in coordinates, it gives the classical
expression of the Hamilton equations:
$$
\frac{dq_i^A}{dt} = \derpar{h}{p^i_A} \circ \sigma \quad , \quad \frac{dp^i_A}{dt} = -\derpar{h}{q_i^A} \circ \sigma \ .
$$

For the case of singular higher-order Lagrangian systems,
in general there is no way to associate a canonical Hamiltonian formalism,
unless some minimal regularity condition are imposed \cite{art:Gracia_Pons_Roman91}.
In particular:

\begin{definition}
A Lagrangian $\Lag \in \Cinfty(\Tan^kQ)$ is said to be an
{\rm almost-regular Lagrangian function} of order $k$ if:
\begin{enumerate}
\item 
$\Leg(\Tan^{2k-1}Q) = P_o$ is a closed submanifold of $\Tan^*(\Tan^{k-1}Q)$.

(We denote the natural embedding by $j_{P_o} \colon P_o \hookrightarrow \Tan^*(\Tan^{k-1}Q)$).
\item
$\Leg$ is a surjective submersion on its image.
\item
For every $p \in \Tan^{2k-1}Q$, the fibers $\Leg^{-1}(\Leg(p))$
are connected submanifolds of $\Tan^{2k-1}Q$.
\end{enumerate}
Then $(\Tan^{2k-1}Q,\Lag)$  is an {\rm almost-regular Lagrangian system} of order $k$.
\end{definition}

Denoting the map defined by $\Leg = j_{P_o} \circ \Leg_o$
by $\Leg_o \colon \Tan^{2k-1}Q \to P_o$, we have that
the Lagrangian energy $E_\Lag$ is a $\Leg_o$-projectable function, and then
there is a unique function $h_o\in \Cinfty(P_o)$ such that $\Leg_o^*h_o = E_L$
(see \cite{art:Gracia_Pons_Roman91}). 

This $h_o$ is the {\sl canonical Hamiltonian function} of the almost-regular Lagrangian system
and, taking $\omega_o = j_{P_o}^*\omega_{k-1}$, the triad $(P_o,\omega_o,h_o)$
is the canonical Hamiltonian system associated to the almost regular Lagrangian system
$(\Tan^{2k-1}Q,\Lag)$.
For this system we have the Hamilton equation
\beq
\label{sub0}
\inn(X_{h_o})\omega_o = \d h_o \quad , \quad X_{h_o}\in\vf(P_o) \ .
\eeq

As $\omega_o$ is, in general, a presymplectic form, in the best cases,
this equation has some vector field $X_{h_o}$ solution only on the points of some submanifold
$P_f\hookrightarrow P_o\hookrightarrow \Tan^*(\Tan^{k-1}Q)$, for which $X_{h_o}$ is tangent to $P_f$.
This vector field is not unique, in general.
It can be proved that $P_f = \Leg(S_f)$,
where $S_f\hookrightarrow\Tan^{2k-1}Q$ is the submanifold where
there are vector field solutions to the Lagrangian equation (\ref{eqn:Cap03_IntrinsicLagEq})
which are tangent to $S_f$ (see the above section).
Furthermore, as $\Leg_o$ is a submersion, for every vector field $X_{h_o} \in \vf(\Tan^*(\Tan^{k-1}Q))$
which is a solution to the Hamilton equation (\ref{sub0}) on $P_f$, and tangent to $P_f$,
there exists some semispray of type $1$, $X_\Lag \in \vf(\Tan^{2k-1}Q)$, 
which is a solution of the Euler-Lagrange equation on $S_f$, and tangent to $S_f$,
such that ${\Leg_o}_*X_\Lag= X_{h_o}$.
This $\Leg_o$-projectable semispray of type $1$ could be defined only on the
points of another submanifold $M_f\hookrightarrow S_f$.
(See \cite{art:Gracia_Pons_Roman91,art:Gracia_Pons_Roman92} for a detailed exposition of all these topics).

\section{Skinner-Rusk unified formalism}
\label{SkinnerRusk}

\subsection{Unified phase space. Geometric and dynamical structures}

Let $\Lag \in \Cinfty(\Tan^{k}Q)$ be the Lagrangian function of order $k$ of the system.
First we construct the {\sl unified phase space}
$$
\W = \Tan^{2k-1}Q \times_{\Tan^{k-1}Q} \Tan^*(\Tan^{k-1}Q) 
$$
(the fiber product of the above bundles), which
is endowed with the canonical projections
$$
\pr_1 \colon \Tan^{2k-1}Q \times_{\Tan^{k-1}Q} \Tan^*(\Tan^{k-1}Q) \to \Tan^{2k-1}Q
\quad ; \quad
\pr_2 \colon \Tan^{2k-1}Q \times_{\Tan^{k-1}Q} \Tan^*(\Tan^{k-1}Q) \to \Tan^*(\Tan^{k-1}Q) \ ,
$$
and also with the canonical projections onto $\Tan^{k-1}Q$.
So we have the diagram:
$$
\xymatrix{
\ & \W \ar[dl]_-{\pr_1} \ar[dr]^-{\pr_2} & \ \\
\Tan^{2k-1}Q \ar[dr]_-{\rho^{2k-1}_{k-1}} \ar@/_1.3pc/[ddr]_-{\beta^{2k-1}} & \ &
 \Tan^*(\Tan^{k-1}Q) \ar[dl]^-{\pi_{\Tan^{k-1}Q}} \\
\ & \Tan^{k-1}Q \ar[d]_-{\beta^{k-1}} & \ \\
\ & Q & \
}
$$

If $(U,q_0^A)$ is a local chart of coordinates in $Q$, denoting by 
$((\beta^{2k-1})^{-1}(U);q_0^A,q_1^A,\ldots,q_{2k-1}^A)$ and
$((\pi_{\Tan^{k-1}Q}\circ \beta^{k-1})^{-1}(U);q_0^A,q_1^A,\ldots,q_{k-1}^A,p^0_A,p^0_A,\ldots,p^{k-1}_A)$
the induced charts in $\Tan^{2k-1}Q$ and in $\Tan^*(\Tan^{k-1}Q)$, respectively,
we have that $(q_0^A,\ldots,q_{k-1}^A;q_{k}^A,\ldots,q_{2k-1}^A;p^0_A,\ldots,p^{k-1}_A)$
are the natural coordinates in the suitable open domain $W\subset\W$.
Note that $\dim(\W) = 3kn$.

The bundle $\W$ is endowed with some canonical geometric structures.
First, let $\omega_{k-1}\in\df^2(\Tan^*(\Tan^{k-1}Q))$ be the canonical symplectic form
 of $\Tan^*(\Tan^{k-1}Q)$. Then we define
$$
\Omega = \pr_2^*\omega_{k-1} \in \df^2(\W) \ ,
$$
which is a presymplectic form in $\W$, whose local expression is
\begin{equation}\label{eqn:Cap06_LocalCoordOmega}
\Omega = \pr_2^*\omega_{k-1} = 
\pr_2^*\left(\d q_i^A \wedge \d p^i_A\right) = \d q_i^A \wedge \d p^i_A \ .
\end{equation}
Observe that 
\beq
\label{eqn:Cap06_LocalBasisKerOmega}
\ker\,\Omega = \left\langle \derpar{}{q^k},\ldots,\derpar{}{q^{2k-1}} \right\rangle= \vf^{V(\pr_2)}(\W) \ .
\eeq

The second relevant canonical structure in $\W$ is the following:

\begin{definition}
Let $p \in \Tan^{2k-1}Q$, its projection $q = \rho^{2k-1}_{k-1}(p)$ to $\Tan^{k-1}Q$, and a
covector $\alpha_q \in \Tan_q^*(\Tan^{k-1}Q)$. 
The {\rm coupling function} $\C \in \Cinfty(\W)$ is defined as follows:
\begin{equation}\label{eqn:Cap06_DefCouplingFunc}
\begin{array}{rcl} \C \colon \Tan^{2k-1}Q \times_{\Tan^{k-1}Q} \Tan^*(\Tan^{k-1}Q) 
& \longrightarrow & \R \\ (p,\alpha_q) & \longmapsto & \langle \alpha_q \mid j_{k}(p)_q \rangle \end{array} \ ,
\end{equation}
where $j_{k} \colon \Tan^{2k-1}Q \to \Tan(\Tan^{k-1}Q)$
is the canonical injection introduced in (\ref{eqn:Cap02_DefCanonicalImmersion}),
 $j_{k}(p)_q$ is the corresponding tangent vector to $\Tan^{k-1}Q$ in $q$,
 and $\langle \alpha_q \mid j_{k}(p)_q \rangle \equiv \alpha_q(j_{k}(p)_q)$
 denotes the canonical pairing between vectors of  $\Tan_q(\Tan^{k-1}Q)$ and
 covectors of $\Tan^*_q(\Tan^{k-1}Q)$.
\end{definition}

Note that, in this case, $j_k \colon \Tan^{2k-1}Q \to \Tan(\Tan^{k-1}Q)$ is a diffeomorphism.

In local coordinates, if $p = (q_0^A,\ldots,q_{k-1}^A,q_k^A,\ldots,q_{2k-1}^A)$,
then $q = \rho^{2k-1}_{k-1}(p) = (q_0^A,\ldots,q_{k-1}^A)$,
and bearing in mind the local expression (\ref{eqn:Cap02_LocalCoordCanonicalImmersion}) of $j_{k}$,
we have $j_{k}(p) = (q_0^A,\ldots,q_{k-1}^A,q_1^A,\ldots,q_k^A)$. Therefore if
$\ds j_{k}(p)_q = q_{i+1}^A\restric{\derpar{}{q_i^A}}{q} \in \Tan_q(\Tan^{k-1}Q)$,
and if 
$\ds \alpha_q = p^i_A \restric{\d q_i^A}{q}$
we obtain the following local expression for the coupling function $\C$
\begin{equation}\label{eqn:Cap06_LocalCoordCouplingFunc}
\C(p,\alpha_q) = \langle \alpha_q \mid j_{k}(p)_q \rangle = 
\left\langle p^i_A \restric{dq_i^A}{q} \bigg| \, q_{i+1}^A
\restric{\derpar{}{q_i^A}}{q} \right\rangle = \restric{p^i_Aq_{i+1}^A}{q} \ .
\end{equation}
Observe that, if $k=1$, the map $j_1 \colon \Tan Q \to \Tan Q$ is the identity on $\Tan Q$,
and we recover the standard canonical coupling between vectors in $\Tan_pQ$ and covectors in $\Tan^*_pQ$.

Using the coupling function, given a Lagrangian function
$\Lag\in\Cinfty(\Tan^kQ)$, we can define the
{\sl Hamiltonian function} $H \in \Cinfty(\W)$ as
\begin{equation}\label{eqn:Cap06_DefUnifiedHamiltFunc}
H = \C - (\rho^{2k-1}_k\circ\pr_1)^*\Lag \ ,
\end{equation}
whose coordinate expression is
\begin{equation}
\label{eqn:Cap06_LocalCoordUnifiedHamiltFunc}
H = p^i_Aq_{i+1}^A - \Lag(q_0^A,\ldots,q_k^A) \ .
\end{equation}

Now, $(\W,\Omega,H)$ is a presymplectic Hamiltonian system.

Finally, in order to give a complete description of the dynamics of higher-order Lagrangian systems,
we need to introduce the following concept:

\begin{definition}
A vector field $X\in\vf(\W)$ is said to be a {\rm semispray of type $r$} in $\W$ if,
for every integral curve $\sigma \colon I \subset\R \to \W$ of $X$,
the curve $\sigma_1 = \pr_1 \circ \sigma \colon I \to \Tan^{2k-1}Q$ satisfies that, if $\gamma = \beta^{2k-1}\circ\sigma_1$,
$\tilde{\gamma}^{2k-r} = \rho^{2k-1}_{2k-r}\circ\sigma_1$.

In particular, $X\in\vf(\W)$ is a {\rm semispray of type $1$}
if $\tilde{\gamma}^{2k-1} = \sigma_1$.
\end{definition}
The local expression of a semispray of type $r$ in $\W$ is
$$
X = \sum_{i=0}^{2k-1-r}q_{i+1}^A\derpar{}{q_i^A} + \sum_{i=2k-r}^{2k-1}X_i^A\derpar{}{q_i^A}
+\sum_{i=0}^{k-1}G^i_A\derpar{}{p^i_A} \ ,
$$
and, in particular, for a semispray of type $1$ in $\W$ we have
$$
X = \sum_{i=0}^{2k-2}q_{i+1}^A\derpar{}{q_i^A} + X_{2k-1}^A\derpar{}{q_{2k-1}^A}
+\sum_{i=0}^{k-1}G^i_A\derpar{}{p^i_A} \ .
$$

\subsection{Dynamical vector fields}

\subsubsection{Dynamics in $\W = \Tan^{2k-1}Q \times_{\Tan^{k-1}Q} \Tan^*(\Tan^{k-1}Q)$}
\label{section:DynamicsW}

As we know, the dynamical equation of the presymplectic Hamiltonian system $(\W,\Omega,H)$
is geometrically written as
\begin{equation}\label{eqn:Cap06_EqDinImp}
\inn(X)\Omega = \d H \quad ; \quad X \in \vf(\W) \ .
\end{equation}
Then, according to \cite{art:Gotay_Nester_Hinds78} we have:

\begin{prop}
\label{prop:Cap06_ExistSolEqDin}
Given the presymplectic Hamiltonian system $(\W, \Omega,H)$,
a solution $X \in \vf(\W)$ to equation (\ref{eqn:Cap06_EqDinImp})
exists only on the points of the submanifold $\W_c \hookrightarrow \W$ defined by
\beq
\label{W0}
\W_c = \left\{ p \in \W \colon \xi(p) \equiv (\inn(Y)\d H)(p) = 0 \ , \ \forall \, Y \in \ker\,\Omega \right\} \ .
\eeq
\end{prop}

We have the following result:

\begin{prop}
\label{prop:Cap06_W0GrafFL}
The submanifold $\W_c \hookrightarrow \W$ contains a submanifold
$\W_o \hookrightarrow \W_c$ which is the graph of the 
Legendre-Ostrogradsky map defined by $\Lag$; that is, $\W_o = {\rm graph}\,\Leg$.
\end{prop}
\begin{proof}
As $\W_c$ is defined by (\ref{W0}), it suffices to prove that
the constraints defining $\W_c$ give rise to those defining the graph of the 
Legendre-Ostrogradsky map associated to $\Lag$.
We make this calculation in coordinates.
Taking the local expression (\ref{eqn:Cap06_LocalCoordUnifiedHamiltFunc})
of the Hamiltonian function $H \in \Cinfty(\W)$, we have 
$$
\d H = \sum_{i=0}^{k-1}(q_{i+1}^A\d p^i_A + p^i_A\d q_{i+1}^A) - \sum_{i=0}^{k} \derpar{\Lag}{q_i^A}\d q_i^A \ ,
$$
and using the local basis of $\ker\,\Omega$ given in (\ref{eqn:Cap06_LocalBasisKerOmega}),
we obtain that the equations defining the submanifold $\W_c$ are
$$
\inn(Y)\d H = 0 \Longleftrightarrow p^{k-1}_A - \derpar{\Lag}{q_k^A} = 0 \ .
$$
Observe that these expressions relate the momentum coordinates
$p^{k-1}_A$ with the Jacobi-Ostrogadsky functions
$\ds \hat p^{k-1}_A= \partial \Lag / \partial q_k^A$, and so we obtain
the last group of equations of the Legendre-Ostrogradsky map.
Furthermore, in Section \ref{subsection:hamform} we have seen that the other Jacobi-Ostrogradsky functions
$\hat p^{r-1}_A$ ($1\leq r\leq k-1$) satisfy the relations (\ref{eqn:Cap04_MomentumCoordRelation}).
Thus we can consider that $\W_c$ contains a submanifold $\W_o$ which
can be identified with the graph of a map
$$
\begin{array}{rcl} F \colon \Tan^{2k-1}Q & \longrightarrow & \Tan^*(\Tan^{k-1}Q) \\ 
(q_i^A) & \longmapsto & (q_0^A,\ldots,q_{k-1}^A,p^0_A,\ldots,p^{k-1}_A) 
\end{array}
$$
which we identify with the Legendre-Ostrogradsky map
by making the identification $p^{r-1}_A=\hat p^{r-1}_A$.
\end{proof}

\textbf{Remark}:
The submanifold $\W_o$ can be obtained from $\W_c$ using a constraint algorithm.
Hence, $\W_o$ acts as the initial phase space of the system.

We denote by $j_o \colon \W_o \hookrightarrow \W$ the natural embedding
and by $\vf_{\W_o}(\W)$ the set of vector fields in $\W$ at support on $\W_o$.
Hence, we look for vector fields $X\in\vf_{\W_o}(\W)$ which are solutions
to equation (\ref{eqn:Cap06_EqDinImp}) at support on $\W_o$; that is
\begin{equation}\label{eqn:Cap06_EqDinSupW0}
\restric{\left[\inn(X)\Omega - \d H\right]}{\W_o} = 0 \ .
\end{equation}

In natural coordinates a generic vector field in $\vf(\W)$ is
$$
X = \sum_{i=0}^{k-1}f_i^A\derpar{}{q_i^A} + \sum_{i=k}^{2k-1}F_i^A\derpar{}{q_i^A} +
 \sum_{i=0}^{k-1}G^i_A\derpar{}{p^i_A} \ ,
$$
bearing in mind the local expressions of $\Omega$ and $\d H$, from (\ref{eqn:Cap06_EqDinImp}),
we obtain the following system of $(2k+1)n$ equations
\bea
 f_i^A  = q_{i+1}^A \ , \label{eqn:Cap06_SODE}\\
 G^0_A  = \displaystyle \derpar{\Lag}{q_0^A} \quad , \quad
 G^i_A = \displaystyle \derpar{\Lag}{q_i^A} -p^{i-1}_A = d_T(p^i_A) \ ,
  \label{eqn:Cap06_EqDin} \\
p^{k-1}_A - \derpar{\Lag}{q_k^A} = 0 \ , \label{eqn:Cap06_FL}
\eea
where $0 \leqslant i \leqslant k-1$ in (\ref{eqn:Cap06_SODE}) and $1 \leqslant i \leqslant k-1$ in (\ref{eqn:Cap06_EqDin}).
Therefore
\beq
\label{Xcoor}
X = \sum_{i=0}^{k-1}q_{i+1}^A\derpar{}{q_i^A} + \sum_{i=k}^{2k-1}F_i^A\derpar{}{q_i^A} + 
\derpar{\Lag}{q_0^A}\derpar{}{p^0_A} + \sum_{i=1}^{k-1}d_T(p^i_A)\derpar{}{p^i_A} \ .
\eeq
We can observe that equations (\ref{eqn:Cap06_FL}) are just a compatibility condition
that, together with the other conditions for the momenta,
say that the vector fields $X$ exist only with support on the submanifold
defined by the graph of the Legendre-Ostrogradsky map. So we recover,
in coordinates, the result stated in Propositions \ref{prop:Cap06_ExistSolEqDin} and \ref{prop:Cap06_W0GrafFL}.
Furthermore, this local expression shows that $X$ is a semispray of type $k$ in $\W$.

The component functions $F_i^A$, $k \leqslant i \leqslant 2k-1$, are undetermined.
Nevertheless, we must study the tangency of $X$ to the submanifold $\W_o$;
that is, we have to impose that $\restric{\Lie(X)\xi}{\W_o} \equiv \restric{X(\xi)}{\W_o} = 0$,
for every constraint function $\xi$ defining $\W_o$.
So, taking into account Prop. \ref{prop:Cap06_W0GrafFL},
these conditions lead to
\begin{align*}
&\left(\sum_{i=0}^{k-1}q_{i+1}^A\derpar{}{q_i^A} + \sum_{i=k}^{2k-1}F_i^A\derpar{}{q_i^A} + 
\derpar{\Lag}{q_0^A}\derpar{}{p^0_A} + \sum_{i=1}^{k-1}d_T(p^i_A)\derpar{}{p^i_A}\right) \left( p^{k-1}_A -
 \derpar{\Lag}{q_k^A} \right) = 0 \\
&\left(\sum_{i=0}^{k-1}q_{i+1}^A\derpar{}{q_i^A} + \sum_{i=k}^{2k-1}F_i^A\derpar{}{q_i^A} + 
\derpar{\Lag}{q_0^A}\derpar{}{p^0_A} + \sum_{i=1}^{k-1}d_T(p^i_A)\derpar{}{p^i_A}\right) \left( p^{k-2}_A -
 \sum_{i=0}^{1}(-1)^i d_T^i\left(\derpar{\Lag}{q_{k-1+i}^A}\right)  \right) = 0 \\
&\qquad \qquad \qquad \qquad \vdots \\
&\left(\sum_{i=0}^{k-1}q_{i+1}^A\derpar{}{q_i^A} + \sum_{i=k}^{2k-1}F_i^A\derpar{}{q_i^A} + 
\derpar{\Lag}{q_0^A}\derpar{}{p^0_A} + \sum_{i=1}^{k-1}d_T(p^i_A)\derpar{}{p^i_A}\right) \left( p^{1}_A - 
\sum_{i=0}^{k-2}(-1)^i d_T^i\left(\derpar{\Lag}{q_{2+i}^A}\right)  \right) = 0 \\
&\left(\sum_{i=0}^{k-1}q_{i+1}^A\derpar{}{q_i^A} + \sum_{i=k}^{2k-1}F_i^A\derpar{}{q_i^A} + 
\derpar{\Lag}{q_0^A}\derpar{}{p^0_A} + \sum_{i=1}^{k-1}d_T(p^i_A)\derpar{}{p^i_A}\right) \left( p^{0}_A - 
\sum_{i=0}^{k-1}(-1)^i d_T^i\left(\derpar{\Lag}{q_{1+i}^A}\right)  \right) = 0 \ ,
\end{align*}
and, from here, we obtain the following $kn$ equations
\begin{equation}
\label{eqn:Cap06_TanVectFieldX}
\begin{array}{l}
\displaystyle \left(F_k^B-q_{k+1}^B\right)\derpars{\Lag}{q_k^B}{q_k^A} = 0 \\[10pt]
\displaystyle \left(F_{k+1}^B - q_{k+2}^B\right)\derpars{\Lag}{q_k^B}{q_k^A} - 
\left(F_k^B-q_{k+1}^B \right) d_T\left(\derpars{\Lag}{q_k^B}{q_k^A}\right) = 0 \\
\qquad \qquad \qquad \qquad \vdots \\
\displaystyle \left(F_{2k-2}^B - q_{2k-1}^B\right)\derpars{\Lag}{q_k^B}{q_k^A} -
 \sum_{i=0}^{k-3} \left(F_{k+i}^B-q_{k+i+1}^B \right) (\cdots\cdots) = 0 \\
\displaystyle (-1)^k\left(F_{2k-1}^B - d_T\left(q_{2k-1}^B\right)\right) \derpars{\Lag}{q_k^B}{q_k^A} + 
\sum_{i=0}^{k} (-1)^id_T^i\left( \derpar{\Lag}{q_i^A} \right) - \sum_{i=0}^{k-2} \left(F_{k+i}^B-q_{k+i+1}^B \right)
 (\cdots\cdots) = 0 \ ,
\end{array}
\end{equation}
where the terms in brackets $(\cdots\cdots)$ contain relations involving partial derivatives
of the Lagrangian and applications of $d_T$ which for simplicity are not written.
These are just the Lagrangian equations for the components of $X$,
as we have seen in (\ref{eqn:Cap03_LocalCoordLagEq}).
These equations can be compatible or not,
and a sufficient condition to ensure compatibility is the regularity of the
Lagrangian function.  In particular, we have:

\begin{prop}
\label{prop:Cap06_RegLag}
If $\Lag \in \Cinfty(\Tan^kQ)$ is a regular Lagrangian function, then
there exists a unique vector field $X \in \vf_{\W_o}(\W)$ which is a solution
to equation (\ref{eqn:Cap06_EqDinSupW0}); it  is  tangent to $\W_o$,
and is a semispray of type $1$ in $\W$.
\end{prop}
\begin{proof}
As the Lagrangian function $\Lag$ is regular, the Hessian matrix
$\ds \left(\derpars{\Lag}{q_k^B}{q_k^A}\right)$ is regular at every point,
and this allows us to solve the above $k$ systems of $n$ equations (\ref{eqn:Cap06_TanVectFieldX})
determining all the functions $F_i^A$ uniquely, as follows
\bea
\label{eqn:Cap06_LagEqReg}
 F_i^A = q_{i+1}^A \quad , \quad  (k \leqslant i \leqslant 2k-2) \\
 (-1)^k\left(F_{2k-1}^B - d_T\left(q_{2k-1}^B\right)\right) \derpars{\Lag}{q_k^B}{q_k^A} +
 \sum_{i=0}^{k} (-1)^id_T^i\left( \derpar{\Lag}{q_i^A} \right) = 0 \ .
 \nonumber
\eea
In this way, the tangency condition holds for $X$ at every point on $\W_o$.
Furthermore, the equalities (\ref{eqn:Cap06_LagEqReg}) show that
$X$ is a semispray of type $1$ in $\W$
\end{proof}

However, if $\Lag$ is  not regular, the equations (\ref{eqn:Cap06_TanVectFieldX})
can be compatible or not. In the most favourable cases,
there is a submanifold $\W_f \hookrightarrow \W_o$ (it could be $\W_f = \W_o$)
such that there exist vector fields $X\in\vf_{\W_o}(\W)$, tangent to $\W_f$,
which are solutions to the equation
\begin{equation}
\label{eqn:Cap06_EqDinSupWf}
\restric{\left[\inn(X)\Omega - dH\right]}{\W_f} = 0 \ .
\end{equation}
In this case, the equations (\ref{eqn:Cap06_TanVectFieldX})
are not compatible, and the compatibility condition gives rise to new constraints.

\subsubsection{Dynamics in $\Tan^{2k-1}Q$}

Now we study how to recover the Lagrangian dynamics from
the dynamics in the unified formalism, using the dynamical vector fields.

First we have the following results:

\begin{prop}
\label{prop:Cap06_pr1Difeo}
The map $\overline{\pr}_1 = \pr_1 \circ j_o \colon \W_o \to \Tan^{2k-1} Q$ is a diffeomorphism.
\end{prop}
\begin{proof}
As $\W_o = {\rm graph}\,\Leg$, we have that $\Tan^{2k-1} Q \simeq \W_o$.
Furthermore, $\overline{\pr}_1$ is a surjective submersion and, by the equality between dimensions,
it is also an injective immersion and hence it is a diffeomorphism.
\end{proof}

\begin{lem}
\label{lemma:Cap06_LagForm}
If $\omega_{k-1} \in \df^2(\Tan^*(\Tan^{k-1}Q))$ is the canonical symplectic $2$-form in 
$\Tan^*(\Tan^{k-1}Q)$, and $\omega_{\Lag} = \Leg^*\omega_{k-1}$ is the
Lagrangian $2$-form,  then $\Omega = \pr_1^*\omega_\Lag$.
\end{lem}
\begin{proof}
In fact,
$$
\pr_1^*\omega_{\Lag} = \pr_1^*(\Leg^*\omega_{k-1})=(\Leg \circ \pr_1)^*\omega_{k-1}=
\pr_2^*\omega_{k-1}= \Omega \ .
$$
\end{proof}

\begin{lem}
\label{lemma:Cap06_LagEnergy}
There exists a unique function $E_\Lag \in \Cinfty(\Tan^{2k-1} Q)$ such that $\pr_1^*E_\Lag=H$.
This function $E_\Lag$ is the Lagrangian energy.
\end{lem}
\begin{proof}
As $\overline{\pr}_1$ is a diffeomorphism, we can define the function
$E_\Lag=(\overline{\pr}_1^{-1}\circ j_o)^*H \in \Cinfty(\Tan^{2k-1} Q)$,
which obviously verifies that $\pr_1^*E_\Lag=H$. 

In order to prove that $E_\Lag$ is the Lagrangian energy defined previously, we calculate its local expression in coordinates.
Thus, from (\ref{eqn:Cap06_LocalCoordUnifiedHamiltFunc}) we obtain that
$$
\overline{\pr}_1^*E_\Lag = H=\sum_{i=0}^{k-1} p^i_Aq_{i+1}^A - \Lag(q_0^A,\ldots,q_k^A) \ ,
$$
but $\W_o\hookrightarrow \W$ is the graph of the Legendre-Ostrogradsky map, and by Prop.
\ref{prop:Cap06_W0GrafFL} we have
$\ds p^i_A = \sum_{j=0}^{k-i-1}(-1)^j d_T^j\left(\derpar{\Lag}{q_{i+1+j}^A}\right)$, and then
\begin{align*}
\overline{\pr}_1^*E_\Lag &= \sum_{i=0}^{k-1}\sum_{j=0}^{k-i-1} q_{i+1}^A(-1)^j d_T^j\left(\derpar{\Lag}{q_{i+1+j}^A}\right) - 
\Lag(q_0^A,\ldots,q_k^A) \\
&= \sum_{i=1}^k \sum_{j=0}^{k-i}q_i^A (-1)^jd_T^j\left(\derpar{\Lag}{q_{i+j}^A}\right) - \Lag(q_0^A,\ldots,q_k^A) \ .
\end{align*}
Now, as $\overline{\pr}_1 = \pr_1 \circ j_o$ and $\pr_1^*q_i^A = q_i^A$, we obtain finally
\begin{equation*}
E_\Lag = \sum_{i=1}^k\sum_{j=0}^{k-i} q_i^A (-1)^jd_T^j\left(\derpar{\Lag}{q_{i+j}^A}\right) -\Lag(q_0^A,\ldots,q_k^A)
\end{equation*}
which is the local expression (\ref{eqn:Cap03_LocalCoordLagEnergy}) of the Lagrangian energy.
\end{proof}

Using these results, we can recover an Euler-Lagrange vector field in $\Tan^{2k-1}Q$
starting from a vector field $X \in \vf_{\W_0}(\W)$ tangent to $\W_o$, a
solution to (\ref{eqn:Cap06_EqDinSupW0}). First we have:

\begin{lem}
\label{lemma:Cap06_LagVectField}
Let $X \in \vf(\W)$ be a vector field tangent to $\W_o$. 
Then there exists a unique vector field $X_\Lag\in \vf(\Tan^{2k-1}Q)$ such that
$X_\Lag \circ \pr_1 \circ j_o = \Tan\pr_1 \circ X \circ j_o$.
\end{lem}
\begin{proof}
As $X \in \vf(\W)$ is tangent to $\W_o$, there exists a vector field $X_o \in \vf(\W_o)$ such that
$\Tan j_o \circ X_o = X \circ j_o$.
Furthermore, as $\overline{\pr}_1$ is a diffeomorphism, there is a unique vector field
$X_\Lag \in \vf(\Tan^{2k-1}Q)$ which is $\overline{\pr}_1$-related with $X_o$; that is,
$X_\Lag \circ \overline{\pr}_1 = \Tan\overline{\pr}_1 \circ X_o$. Then
$$
X_\Lag \circ \pr_1 \circ j_o = X_\Lag \circ \overline{\pr_1} = \Tan\overline{\pr}_1 \circ X_o = 
\Tan\pr_1 \circ \Tan j_o \circ X_o = \Tan\pr_1 \circ X \circ j_o
$$
\end{proof}

And as a consequence we obtain:

\begin{teor}
\label{thm:Cap06_CorrX-XL}
Let $X \in \vf_{\W_o}(\W)$ be a vector field solution to equation (\ref{eqn:Cap06_EqDinSupW0}) and
 tangent to $\W_o$ (at least on the points of a submanifold $\W_f \hookrightarrow \W_o$).
 Then there exists a unique semispray of type $k$, $X_\Lag \in \vf(\Tan^{2k-1} Q)$, 
 which is a solution to the equation
\begin{equation}
\label{eqn:Cap06_EqLag}
\inn(X_\Lag)\omega_\Lag - dE_\Lag = 0
\end{equation}
(at least on the points of $S_f = \pr_1(\W_f)$). In addition, if $\Lag \in \Cinfty(\Tan^kQ)$
is a regular Lagrangian, then $X_\Lag$ is a semispray of type $1$, and 
hence it is the Euler-Lagrange vector field.

\noindent Conversely, if $X_\Lag \in \vf(\Tan^{2k-1}Q)$ is a semispray of type $k$ (resp., of type $1$),
which is a solution to equation (\ref{eqn:Cap06_EqLag}) (at least on the points of a submanifold
 $S_f \hookrightarrow \Tan^{2k-1}Q$), then there exists a unique vector field $X \in \vf_{\W_o}(\W)$
 which is a solution to equation (\ref{eqn:Cap06_EqDinSupW0}) 
 (at least on $\W_f = \overline{\pr}_1^{-1}(S_f) \hookrightarrow \W_o \hookrightarrow \W$),
 and it is a semispray of type $k$ in $\W$ (resp., of type $1$).
\end{teor}
\begin{proof}
Applying Lemmas \ref{lemma:Cap06_LagForm}, \ref{lemma:Cap06_LagEnergy},
 and \ref{lemma:Cap06_LagVectField}, we have:
$$
0 = \restric{\left[\inn(X)\Omega - \d H\right]}{\W_o}=
\restric{\left[\inn(X)\pr_1^*\omega_\Lag - \d\pr_1^*E_\Lag\right]}{\W_o} =
\pr_1^*\restric{\left[\inn(X_\Lag)\omega_\Lag - \d E_\Lag\right]}{\W_o} \ ,
$$
but, as $\pr_1$ is a surjective submersion, this is equivalent to
$$
0=\restric{\left[\inn(X_\Lag)\omega_\Lag - \d E_\Lag\right]}{\pr_1(\W_o)} = 
\restric{\left[\inn(X_\Lag)\omega_\Lag - \d E_\Lag\right]}{\Tan^{2k-1}Q} = 0 \ ,
$$
since $\pr_1(\W_o) = \Tan^{2k-1}Q$.
The converse is immediate, reversing this reasoning.

In order to prove that $X_\Lag$ is a semispray of type $k$, we proceed in coordinates.
From the local expression (\ref{Xcoor}) for the vector field $X$
(where the functions $F_i^A$ are the solutions of the equations (\ref{eqn:Cap06_TanVectFieldX})),
and using Lemma \ref{lemma:Cap06_LagVectField},
we obtain that the local expression of $X_\Lag \in \vf(\Tan^{2k-1}Q)$ is
$$
X_\Lag = \sum_{i=0}^{k-1}q_{i+1}^A\derpar{}{q_i^A} + \sum_{i=k}^{2k-1}F_i^A\derpar{}{q_i^A}\ ,
$$
and then
$$
 J_k(X_\Lag) = \sum_{i=0}^{k-1}\frac{(k+i)!}{i!}q_{i+1}^A\derpar{}{q_{k+i}^A} = \Delta_k \ ;
 $$
so $X_\Lag$ is a semispray of type $k$ in $\Tan^{2k-1}Q$.

Finally, if $\Lag \in \Cinfty(\Tan^kQ)$ is a regular Lagrangian, equations (\ref{eqn:Cap06_TanVectFieldX})
become (\ref{eqn:Cap06_LagEqReg}), and hence the local expression of $X$ is
$$
X = \sum_{i=0}^{2k-2}q_{i+1}^A\derpar{}{q_i^A} + F_{2k-1}^A\derpar{}{q_{2k-1}^A} + 
\derpar{\Lag}{q_0^A}\derpar{}{p^0_A} + \sum_{i=1}^{k-1}d_T(p^i_A)\derpar{}{p^i_A} \ .
$$
Therefore
$$
X_\Lag = \sum_{i=0}^{2k-2}q_{i+1}^A\derpar{}{q_i^A} + F_{2k-1}^A\derpar{}{q_{2k-1}^A} \ ,
$$
and then $\ds J_1(X_\Lag) = \sum_{i=0}^{2k-2}(i+1)q_{i+1}^A\derpar{}{q_{i+1}^A} = \Delta_1$,
which shows that $X_\Lag$ is a semispray of type $1$ in $\Tan^{2k-1}Q$.
\end{proof}

{\bf Remarks}:
\bit
\item
It is important to point out that, if $\Lag$ is not a regular Lagrangian,
then $X$ is a semispray of type $k$ in $\W$, but not necessarily a semispray of type $1$.
This means that $X_\Lag$ is a Lagrangian vector field, but it is not necessarily an Euler-Lagrange
vector field (it is not a semispray of type $1$,
but just a semispray of type $k$). Thus, for singular Lagrangians,
this must be imposed as an additional condition in order that the
integral curves of $X_\Lag$ verify the Euler-Lagrange equations.
This is a different from the case of first-order dynamical systems ($k=1$),
where this condition ($X_\Lag$ is a semispray of type $1$; that is, a holonomic vector field)
is obtained straightforwardly in the unified formalism.

In general, only in the most interesting cases have we assured the existence of
a submanifold  $\W_f \hookrightarrow \W_o$ 
and vector fields $X \in \vf_{\W_0}(\W)$ tangent to $\W_f$
which are solutions to the equation (\ref{eqn:Cap06_EqDinSupWf}).
Then, considering the submanifold
$S_f=\pr_1(\W_f)\hookrightarrow \Tan^{2k-1}Q$,
in the best cases (see \cite{art:Batlle_Gomis_Pons_Roman88, art:Gracia_Pons_Roman91,art:Gracia_Pons_Roman92}),
we have that those Euler-Lagrange vector fields $X_\Lag$ exist,
perhaps on another submanifold $M_f\hookrightarrow S_f$
where they are tangent, and are solutions to the equation
\beq
\restric{\left[i_{X_\Lag}\omega_\Lag - \d E_\Lag\right]}{M_f} = 0 \ .
\label{finaleq}
\eeq
\item
Observe also that Theorem \ref{thm:Cap06_CorrX-XL} states that there is a one-to-one correspondence
between vector fields $X \in \vf_{\W_o}(\W)$ which are solutions to equation
(\ref{eqn:Cap06_EqDinSupW0}) and $X_\Lag \in \vf(\Tan^{2k-1}Q)$ solutions to
(\ref{eqn:Cap06_EqLag}), but not uniqueness, unless $\Lag$ is regular. In fact:
\eit

\begin{corol}
\label{corol:Cap06_RegLag}
If $\Lag \in \Cinfty(\Tan^kQ)$ is a regular Lagrangian, then there is a unique
$X \in \vf_{\W_o}(\W)$ tangent to $\W_o$ which is a solution to equation (\ref{eqn:Cap06_EqDinSupW0}),
and it is a semispray of type $1$.
\end{corol}
\begin{proof}
As $\Lag$ is regular, by Proposition \ref{prop:Cap03_LagVectFieldRegLag}
there is a unique semispray of type $1$, $X_\Lag \in \vf(\Tan^{2k-1} Q)$
which is a solution to equation (\ref{eqn:Cap06_EqLag}) on $\Tan^{2k-1}Q$.
Then, by Theorem \ref{thm:Cap06_CorrX-XL}, there is a unique $X \in \vf_{\W_o}(\W)$,
tangent to $\W_o$, which is a solution to (\ref{eqn:Cap06_EqDinSupW0}) on $\W_o$.
\end{proof}

\subsubsection{Dynamics in $\Tan^*(\Tan^{k-1}Q)$}

\paragraph{Hyperregular and regular Lagrangians}\

In order to recover the Hamiltonian formalism, we distinguish between the regular and non-regular cases.
We start with the regular case, although by simplicity we analyze the hyperregular case
(the regular case is recovered from this by restriction on the corresponding open sets
where the Legendre-Ostrogadsky map is a local diffeomorphism).
For this case we have the following commutative diagram
$$
\xymatrix{
\ & \Tan\W \ar@/_/[ddl]_-{\Tan\pr_1} \ar@/^/[ddr]^-{\Tan\pr_2} & \ \\
\ & \Tan\W_o \ar[dl]_-{\Tan\overline{\pr}_1}|<<<<{\hole} \ar[dr]^-{\Tan\overline{\pr}_2} \ar@{^{(}->}[u]_-{\Tan j_o} & \ \\
\Tan(\Tan^{2k-1} Q) & \ & \Tan(\Tan^*(\Tan^{k-1}Q)) \\
\ & \W \ar[ddl]_-{\pr_1} \ar[ddr]^-{\pr_2}|<<<<{\hole} \ar@/^1.95pc/[uuu]^(.35){X} & \ \\
\ & \W_o = {\rm graph}\,{\Leg} \ar[dl]^-{\overline{\pr}_1} \ar[dr]_-{\overline{\pr}_2} \ar@{^{(}->}[u]^-{j_o} 
\ar@/_1.75pc/[uuu]_-{X_o} & \ \\
\Tan^{2k-1}Q \ar[dr]_-{\rho^{2k-1}_{k-1}} \ar[rr]^-{\Leg} \ar[uuu]^-{X_\Lag} \ar@/_1.3pc/[ddr]_-{\beta^{2k-1}} 
& \ & \Tan^*(\Tan^{k-1}Q) \ar[dl]^{\pi_{\Tan^{k-1}Q}} \ar[uuu]_-{X_h} \\
\ & \Tan^{k-1}Q \ar[d]_-{\beta^{k-1}} & \ \\
\ & Q & \ \\
}
$$

\begin{teor}
\label{thm:Cap06_CorrX-Xh}
Let $\Lag \in \Cinfty(\Tan^kQ)$ be a hyperregular Lagrangian, $h \in \Cinfty(\Tan^*(\Tan^{k-1}Q))$
the Hamiltonian function such that $\Leg^*h = E_\Lag$, and
$X \in \vf_{\W_o}(\W)$ the vector field solution to the equation (\ref{eqn:Cap06_EqDinSupW0}),
tangent to $\W_o$. Then, there exists a unique vector field $X_h = \Leg_*X_\Lag \in \vf(\Tan^*(\Tan^{k-1}Q))$
which is a solution to the equation
\begin{equation}
\label{eqn:Cap06_EqHam}
\inn(X_h)\omega_{k-1} - \d h = 0
\end{equation}
Conversely, if $X_h \in \vf(\Tan^*(\Tan^{k-1}Q))$ is a solution to equation (\ref{eqn:Cap06_EqHam}),
then there exists a unique vector field $X \in \vf_{\W_o}(\W)$, tangent to $\W_o$,
which is a solution to equation (\ref{eqn:Cap06_EqDinSupW0}).
\end{teor}
\begin{proof}
If $\Lag$ is hiperregular, then $\overline{\pr}_2 = \Leg \circ \overline{\pr}_1$ is a diffeomorphism,
since it is a composition of diffeomorphisms; then there exists a unique vector field
$X_o \in \vf(\W_o)$ such that ${\overline{\pr}_2}_*X_o = X_h$,
and there is a unique $X \in \vf_{\W_o}(\W)$ such that ${j_o}_*X_o = \restric{X}{\W_o}$.

Now, as $\Leg^*h = E_\Lag$, by applying Lemma \ref{lemma:Cap06_LagEnergy}
we have that
$\pr_1^*(\Leg^*(h)) = \pr_1^*E_\Lag = H$; but
$\Leg \circ \pr_1 = \pr_2$, and then
$\pr_2^*h = H$.  Therefore, by the definition of $\Omega$, we have
 $$
0 =\restric{\left[\inn(X)\Omega - \d H\right]}{\W_o} = 
\restric{\left[\inn(X)\pr_2^*\omega_{k-1} - \d \pr_2^*h\right]}{\W_o} = 
\pr_2^*\restric{\left[\inn(X_h)\omega_{k-1} - \d h\right]}{\W_o} \ .
$$
However, as $\pr_2$ is a surjective submersion and $\pr_2(\W_o) = \Tan^*(\Tan^{k-1}Q)$,
we finally obtain that
$$
\\ 0 = \restric{\left[\inn(X_h)\omega_{k-1} - \d h\right]}{\pr_2(\W_o)}=
\restric{\left[\inn(X_h)\omega_{k-1} - \d h\right]}{\Tan^*(\Tan^{k-1}Q)}
$$
\end{proof}

\paragraph{Singular (almost-regular) Lagrangians} \

Remember that, for almost-regular Lagrangians,
only in the most interesting cases have we assured the existence of
a submanifold  $\W_f \hookrightarrow \W_o$ 
and vector fields $X \in \vf_{\W_0}(\W)$ tangent to $\W_f$
which are solutions to equation (\ref{eqn:Cap06_EqDinSupWf}).
In this case, the dynamical vector fields in the
Hamiltonian formalism cannot be obtained straightforwardly from
the solutions in the unified formalism, but rather by passing through the
Lagrangian formalism and using the Legendre-Ostrogadsky map.

Thus, we can consider the submanifolds
$S_f=\pr_1(\W_f)\hookrightarrow \Tan^{2k-1}Q$
and $P_f=\pr_2(\W_f)= \Leg(S_f)\hookrightarrow \Tan^*(\Tan^{k-1}Q)$.
Then, using Theorem \ref{thm:Cap06_CorrX-XL}, from the vector fields
$X \in \vf_{\W_o}(\W)$ we obtain the corresponding
$X_\Lag \in \vf(\Tan^{2k-1}Q)$, and from these the
semisprays of type $1$ (if they exist) which are perhaps defined
on a submanifold $M_f\hookrightarrow S_f$, are tangent to $M_f$
and are solutions to equation (\ref{finaleq}).
So we have the following commutative diagram
$$
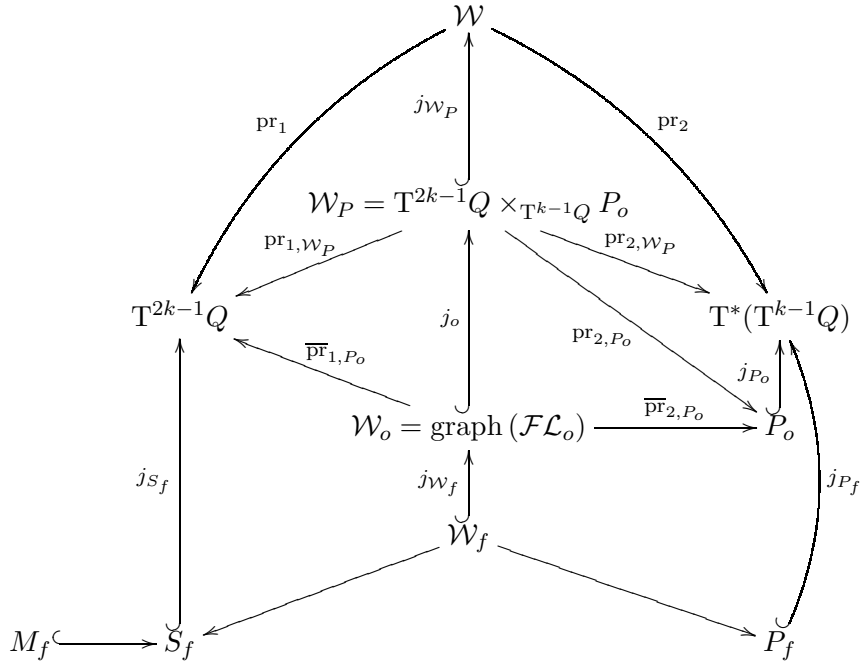
\xymatrix{
\ & \ & \W \ar@/_1.25pc/[dddl]_-{\pr_1} \ar@/^1.25pc/[dddr]^-{\pr_2} & \ \\
\ & \ & \ & \ \\
\ & \ & \W_P = \Tan^{2k-1}Q \times_{\Tan^{k-1}Q} P_o
 \ar@{^{(}->}[uu]^-{j_{\W_P}} \ar[dl]_-{\pr_{1,\W_P}} \ar[dr]^-{\pr_{2,\W_P}} \ar[ddr]_-{\pr_{2,P_o}}  & \ \\
\ & \Tan^{2k-1}Q & \ & \Tan^*(\Tan^{k-1}Q) \\
\ & \ & \W_o = {\rm graph}\,(\Leg_o) \ar@{^{(}->}[uu]^-{j_o} \ar[ul]_-{\overline{\pr}_{1,P_o}} \ar[r]^-{\overline{\pr}_{2,P_o}} 
& P_o \ar@{^{(}->}[u]^{j_{P_o}}  \\
\ & \ & \W_f \ar@{^{(}->}[u]^-{j_{\W_f}} \ar[dl] \ar[dr] & \ \\
M_f \ar@{^{(}->}[r] & S_f \ar@{^{(}->}[uuu]^-{j_{S_f}} & \ & P_f \ar@{^{(}->}@/_1.25pc/[uuu]_-{j_{P_f}}
}
$$

Now, it is proved (\cite{art:Gracia_Pons_Roman92}) that there are Euler-Lagrange
vector fields (perhaps only on the points of another
submanifold $\bar M_f\hookrightarrow M_f$), which are $\Leg$-projectable
on $P_f = \Leg(S_f) \hookrightarrow P_o \hookrightarrow \Tan^*(\Tan^{k-1}Q)$.
These vector fields $X_{h_o}=\Leg_*X_{\Lag} \in \vf(\Tan^*(\Tan^{k-1}Q))$
are tangent to $P_f$ and are solutions to the equation
$$
\restric{\left[ i_{X_{h_o}}\omega_o - \d h_o \right]}{P_f} = 0 \ .
$$

Conversely, as $\Leg_o$ is a submersion, for every solution 
$X_{h_o} \in \vf(\Tan^*(\Tan^{k-1}Q))$ to the last equation,
 there is a semispray of type $1$, 
$X_\Lag \in \vf(\Tan^{2k-1}Q)$, such that ${\Leg_o}_*X_\Lag = X_{h_o}$,
and we can recover solutions to equation (\ref{eqn:Cap06_EqDinSupWf})
using Theorem \ref{thm:Cap06_CorrX-XL}.

\subsection{Integral curves}

After studying the vector fields which are solutions to the dynamical equations,
we analyze their integral curves, showing how to recover the
Lagrangian  and Hamiltonian dynamical trajectories from the
dynamical trajectories in the unified formalism.

Let $X \in \vf_{\W_o}(\W)$ be a vector field tangent to $\W_o$
which is a solution to equation (\ref{eqn:Cap06_EqDinSupW0}), 
and let $\sigma \colon I \subset \R \to \W$ be an integral curve of $X$,
on $\W_o$. As $\tilde{\sigma} = X \circ \sigma$, this means that the following equation holds
\begin{equation}\label{eqn:Cap06_EqDinImpCI}
\inn({\tilde{\sigma}})(\Omega \circ \sigma) = \d H \circ \sigma \ .
\end{equation}
Furthermore, if $\sigma_o \colon I \to \W_o$ is a curve on $\W_o$ such that $j_o \circ \sigma_o = \sigma$,
we have that $\sigma_o$ is an integral curve of the vector field $X_o\in\vf(\W_o)$ associated to $X$,
and $\tilde{\sigma}_o = X_o \circ \sigma_o$.

In local coordinates, if $\sigma(t) = (q_i^A(t),p^j_A(t))$, we have that
\begin{align*}
\dot{q}_i^A(t) = q_{i+1}^A\circ\sigma  & \quad (0 \leqslant i \leqslant k-1) \quad  ; \quad&
\dot{q}_i^A(t) = F_i^A\circ\sigma  \quad (k \leqslant i \leqslant 2k-1) \\
\dot{p}^0_A(t) = \derpar{\Lag}{q_0^A}\circ\sigma  & \qquad \qquad\qquad\qquad\ ; \quad &
\dot{p}^i_A(t) = d_T(p^i_A)\circ\sigma  \quad (1 \leqslant i \leqslant k-1) \ ,
\end{align*}
where $F_i^A$ are solutions to equations (\ref{eqn:Cap06_TanVectFieldX}).

Now, for the Lagrangian dynamical trajectories we have the following result:

\begin{prop}
\label{prop:Cap06_XLCI}
Let $\sigma \colon I \subset\R \to \W$ be an integral curve of a vector field $X$
solution to (\ref{eqn:Cap06_EqDinSupW0}), on $\W_o$.
Then the curve $\sigma_\Lag = \pr_1 \circ \sigma \colon I \to \Tan^{2k-1}Q$ 
is an integral curve of $X_\Lag$.
\end{prop}
\begin{proof}
As $\sigma = j_o \circ \sigma_o$,
using that $\Tan j_o \circ X_o = X \circ j_o$ and that
$\Tan\pr_1 \circ X = X_\Lag \circ \pr_1$, we have
\beann
\tilde{\sigma}_\Lag &=& \widetilde{\pr_1 \circ \sigma} = \widetilde{\pr_1 \circ j_o \circ \sigma_o}=
 \Tan\pr_1 \circ \Tan j_o \circ \tilde{\sigma}_o = \Tan\pr_1 \circ \Tan j_o \circ X_o \circ \sigma_o 
  \\ &=&
   \Tan \pr_1 \circ X \circ j_o \circ \sigma_o = X_\Lag \circ \pr_1 \circ j_o \circ \sigma_o = X_\Lag \circ \sigma_\Lag \ .
\eeann
\end{proof}

\begin{corol}
If $\Lag \in \Cinfty(\Tan^kQ)$ is a regular Lagrangian, then the curve
$\sigma_\Lag = \pr_1 \circ \sigma \colon I \to \Tan^{2k-1}Q$ is the canonical lifting of a curve on $Q$;
 that is, there exists $\gamma \colon I \subset \R \to Q$ such that $\sigma_\Lag = \tilde{\gamma}^{2k-1}$.
\end{corol}
\begin{proof}
It is a straighforward consequence of Proposition \ref{prop:Cap06_XLCI}
and Theorem \ref{thm:Cap06_CorrX-XL}.
\end{proof}

And for the Hamiltonian trajectories, we have:

\begin{prop}
\label{prop:Cap06_XhCI}
Let $\sigma \colon I \subset\R \to \W$ be an integral curve of a vector field $X$
solution to (\ref{eqn:Cap06_EqDinSupW0}), on $\W_o$.
Then the curve $\sigma_h = \Leg \circ \sigma_\Lag \colon I \to \Tan^*(\Tan^{k-1}Q)$ 
is an integral curve of $X_h = \Leg_*(X_\Lag)$.
\end{prop}
\begin{proof}
Given that $\sigma_\Lag$ is an integral curve of $X_\Lag$, Proposition \ref{prop:Cap06_XLCI},
and the definitions of $X_h$ and $\sigma_h$, we obtain
$$
\tilde{\sigma}_h = \widetilde{\Leg \circ \sigma_\Lag} = \Tan\Leg \circ \tilde{\sigma}_\Lag
= \Tan\Leg \circ X_\Lag \circ \sigma_\Lag = X_h \circ \Leg \circ \sigma_\Lag = X_h \circ \sigma_h \ .
$$
Thus $\sigma_h$ is an integral curve of $X_h$.
\end{proof}

The relation among all these integral curves is summarized in the following diagram
$$
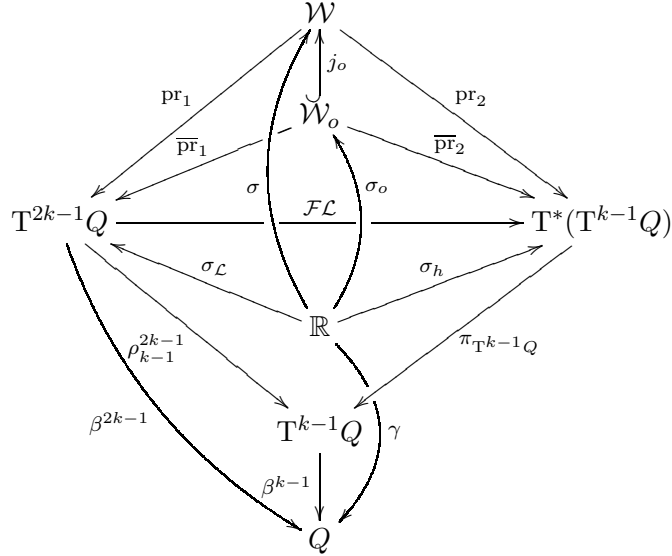
\xymatrix{
\ & \ & \W \ar[ddll]_-{\pr_1} \ar[ddrr]^-{\pr_2} & \ & \ \\
\ & \ & \W_o \ar[dll]_(.45){\overline{\pr}_1}|(.19){\hole} \ar[drr]^-{\overline{\pr}_2} \ar@{^{(}->}[u]_-{j_o} & \ & \ \\
\Tan^{2k-1}Q \ar[ddrr]_-{\rho^{2k-1}_{k-1}} 
\ar@/_1.3pc/[dddrr]_-{\beta^{2k-1}} \ar[rrrr]^-{\Leg}|(.39){\hole}|(.56){\hole} & \ & \ & \ & \Tan^*(\Tan^{k-1}Q) \ar[ddll]^-{\pi_{\Tan^{k-1}Q}} \\
\ & \ & \R \ar@/^1.9pc/[dd]^-{\gamma}|(.31){\hole} \ar@/_1.3pc/[uu]_(.65){\sigma_o} 
\ar@/^1.65pc/[uuu]^(.45){\sigma} \ar[ull]_-{\sigma_\Lag} \ar[urr]^-{\sigma_h} & \ & \ \\
\ & \ & \Tan^{k-1}Q \ar[d]_-{\beta^{k-1}} & \ & \ \\
\ & \ & Q & \ & \
}
$$

{\bf Remark}:
Observe that in Propositions \ref{prop:Cap06_XLCI} and \ref{prop:Cap06_XhCI}
 we make no assumption
on the regularity of the system.
The only considerations in the almost-regular case are that, in general,
the curves are defined in some submanifolds which are determined by the constraint algorithm,
and that $\sigma_\Lag$ is not necessarily the lifting of any curve in $Q$ and this condition
must be imposed. In particular:
\begin{itemize}
\item
If the Lagrangian is regular (or hiperregular), then
${\rm Im}\,(\sigma) \subset \W_o$, ${\rm Im}\,(\sigma_\Lag) \subset \Tan^{2k-1} Q$ 
and ${\rm Im}\, (\sigma_h) \subset \Tan^*(\Tan^{k-1}Q)$.
\item
If the Lagrangian is almost-regular, then
 ${\rm Im}\,(\sigma) \subset \W_f \hookrightarrow \W_o$, 
 ${\rm Im}\,(\sigma_\Lag) \subset S_f \hookrightarrow \Tan^{2k-1} Q$ 
 and ${\rm Im}\,(\sigma_h) \subset P_f \hookrightarrow P_o \hookrightarrow \Tan^*(\Tan^{k-1}Q)$.
\end{itemize}

\section{Examples}
\label{section:examples}

\subsection{The Pais-Uhlenbeck oscillator}

The Pais-Uhlenbeck oscillator is one of the simplest
(regular) systems that can be used to explore the features of higher order
dynamical systems, and has been analyzed in detail in many publications
\cite{art:Pais_Uhlenbeck50,art:Martinez_Montemayor_Urrutia11}.
Here we study it using the unified formalism.

The configuration space for this system is a $1$-dimensional smooth manifold $Q$
with local coordinate $(q_0)$. Taking natural coordinates
in the higher-order tangent bundles over $Q$, the second-order Lagrangian function
$\Lag \in \Cinfty(\Tan^2Q)$ for this system is locally given  by
$$
\Lag(q_0,q_1,q_2) = \frac{1}{2} \left( q_1^2 - \omega^2q_0^2 - \gamma q_2^2 \right) 
$$
where $\gamma$ is some nonzero real constant, and $\omega$ is a real constant.
$\Lag$ is a regular Lagrangian function, since the Hessian matrix of $\Lag$
with respect to $q_2$ is
$$
\left( \derpars{\Lag}{q_2}{q_2} \right) =- \gamma
$$
which has maximum rank, since we assume that $\gamma$ is nonzero.
Notice that, if we take $\gamma = 0$, then $\Lag$ becomes a first-order
regular Lagrangian function, and thus it is a nonsense to study this system using
the higher-order unified formalism.

As this is a second-order dynamical system, the phase space that we consider is
$$
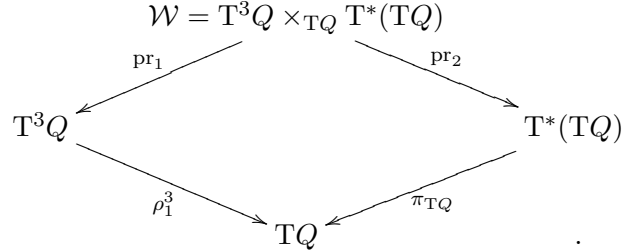
\xymatrix{
\ & \W = \Tan^3Q \times_{\Tan Q} \Tan^*(\Tan Q) \ar[dl]_-{\pr_1} \ar[dr]^-{\pr_2} & \ \\
\Tan^3Q \ar[dr]_-{\rho^{3}_{1}} & \ & \Tan^*(\Tan Q) \ar[dl]^-{\pi_{\Tan Q}} \\
\ & \Tan Q & \  .
}
$$

Denoting the canonical symplectic
form by $\omega_1 \in \df^2(\Tan^*(\Tan Q))$,
we define the presymplectic form $\Omega \in \pr_2^*\omega_1 \in \df^2(\W)$
with the local expression
$$
\Omega = \d q_0 \wedge \d p^0 + \d q_1 \wedge \d p^1\ ,
$$
The Hamiltonian function $H \in \Cinfty(\W)$
in the unified formalism is $H = \C - (\rho^3_2 \circ \pr_1)^*\Lag$, where
$\C$ is the coupling function, whose local expression is
$$
\C(q_0,q_1,q_2,q_3,p^0,p^1) = p^0q_1 + p^1q_2 \ .
$$
and then the Hamiltonian function can be written locally
\begin{equation*}
H(q_0,q_1,q_2,q_3,p^0,p^1) = p^0q_1 + p^1q_2 - \frac{1}{2} \left( q_1^2 - \omega^2q_0^2 - \gamma q_2^2 \right)
\end{equation*}

As stated in the above sections, we can describe the dynamics for this
system in terms of the integral curves of vector fields $X \in \vf(\W)$ which are
solutions to equation (\ref{eqn:Cap06_EqDinImp}). If we take a generic vector field $X$
in $\W$, given locally by
$$
X = f_0 \derpar{}{q_0} + f_1 \derpar{}{q_1} + F_2 \derpar{}{q_2} + F_3 \derpar{}{q_3} +
 G^0\derpar{}{p^0} + G^1\derpar{}{p^1},
$$
taking into acount that
$$
\d H = \omega^2q_0\d q_0 + (p^0-q_1)\d q_1 + (p^1 + \gamma q_2)\d q_2 + q_1 \d p^0 + q_2 \d p^1 \ ,
$$
from the dynamical equation $\inn(X)\Omega = \d H$, we obtain the following system of linear
equations for the coefficients of $X$
\begin{align}
& f_0 = q_1 \label{eqn:Example1_Semispray2_1} \\
& f_1 = q_2 \label{eqn:Example1_Semispray2_2} \\
& G^0 = - \omega^2 q_0 \label{eqn:Example1_VectorFieldG0} \\
& G^1 = q_1 - p^0 \label{eqn:Example1_VectorFieldG1} \\
& p^1 + \gamma q_2 = 0 \label{eqn:Example1_LegTransformation}
\end{align}
Equations (\ref{eqn:Example1_Semispray2_1}) and (\ref{eqn:Example1_Semispray2_2})
give us the condition of semispray of type $2$ for the vector field $X$.
Furthermore, equation (\ref{eqn:Example1_LegTransformation}) is an algebraic
relation stating that the vector field
$X$ is defined along a submanifold $\W_o$ that can be identified with the graph
of the Legendre-Ostrogradsky map, as we have seen in Propositions
\ref{prop:Cap06_ExistSolEqDin} and \ref{prop:Cap06_W0GrafFL}. Thus, using
(\ref{eqn:Example1_Semispray2_1}), (\ref{eqn:Example1_Semispray2_2}),
(\ref{eqn:Example1_VectorFieldG0}) and (\ref{eqn:Example1_VectorFieldG1}),
the vector field $X$ is given locally by
\begin{equation}
\label{eqn:Example1_VectorFieldX}
X = q_1 \derpar{}{q_0} + q_2 \derpar{}{q_1} + F_2 \derpar{}{q_2} + F_3 \derpar{}{q_3} -
 \omega^2q_0\derpar{}{p^0} + \left(q_1 - p^0\right) \derpar{}{p_1} \ .
\end{equation}

As our goal is to recover the Lagrangian and Hamiltonian solutions
from the vector field $X$, we must require $X$ to be a semispray of
type $1$. Nevertheless, as $\Lag$ is a regular Lagrangian function, this
condition is naturally deduced from the formalism, as we have seen in
(\ref{eqn:Cap06_TanVectFieldX}).

Notice that the functions $F_2$ and $F_3$ in (\ref{eqn:Example1_VectorFieldX})
are not determined until the tangency of the vector field $X$ on $\W_o$ is
required. Recall that the Legendre-Ostrogradsky transformation is the map
$\Leg \colon \Tan^3Q \longrightarrow \Tan^*(\Tan Q)$ given in local coordinates by
\beann
\Leg^*p^0 &=&
\derpar{\Lag}{q_1} - d_T\left(\derpar{\Lag}{q_2}\right) \equiv \derpar{\Lag}{q_1} - d_T\left( p^1 \right) = 
q_1 + \gamma q_3 \\ 
\Leg^*p^1 &=& \derpar{\Lag}{q_2} = - \gamma q_2
\eeann
and, as $\gamma \neq 0$, we see that $\Lag$ is a regular Lagrangian since
$\Leg$ is a (local) diffeomorphism. Then, the submanifold
$\W_o = {\rm graph}\,\Leg$ is defined by
$$
\W_o = \left\{ p \in \W \colon \xi_0(p) = \xi_1(p) = 0 \right\} \ ,
$$
where $\xi_r = p^r - \Leg^*p^r$, $r=1,2$. The diagram for this situation is
$$
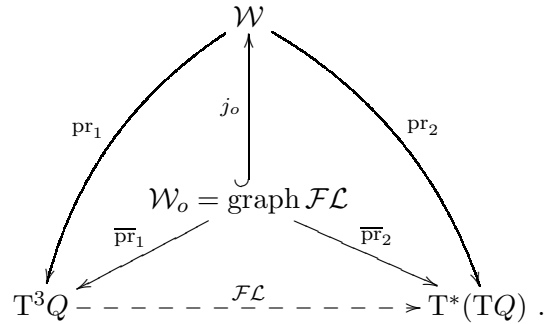
\xymatrix{
\ & \W \ar@/_1.25pc/[dddl]_-{\pr_1} \ar@/^1.25pc/[dddr]^-{\pr_2} & \ \\
\ & \ & \ \\
\ & \W_o = {\rm graph}\,\Leg \ar@{^{(}->}[uu]^-{j_o} \ar[dl]_-
{\overline{\pr}_{1}} \ar[dr]^-{\overline{\pr}_{2}} & \ \\
\Tan^3Q \ar@{-->}[rr]^-{\Leg} & \ & \Tan^*(\Tan Q) \ .
}
$$

Next we compute the tangency condition for $X \in \vf(\W)$ given locally by
(\ref{eqn:Example1_VectorFieldX}) on the submanifold $\W_o \hookrightarrow \W$,
by checking if the following identities hold
$$
\Lie(X)\xi_0\vert_{\W_o} = 0  \quad , \quad \Lie(X)\xi_1\vert_{\W_o} = 0 \ .
$$
As we have seen in Section \ref{section:DynamicsW}, these equations give us the
Lagrangian equations for the vector field $X$; that is, on the points of $\W_o$ we obtain
\begin{align}
\Lie(X)\xi_0 = - \omega^2 q_0 - q_2 - \gamma F_3 = 0
 \label{eqn:Example1_EulerLagrangeVectFieldEq} \\
\Lie(X)\xi_1 = \gamma\left(F_2 - q_3\right) = 0  \ .
\label{eqn:Example1_Semispray1}
\end{align}
Equation (\ref{eqn:Example1_Semispray1}) gives us the condition of semispray
of type $1$ for the vector field $X$ (recall that $\gamma \neq 0$), and equation
(\ref{eqn:Example1_EulerLagrangeVectFieldEq}) is the Euler-Lagrange equation for
the vector field $X$. Notice that, as $\gamma$ is nonzero, these equations
give us a unique solution for $F_2$ and $F_3$.
Thus, there is a unique vector field $X \in \vf(\W)$ solution to the equation
$\restric{\left[ \inn(X)\Omega - \d H \right]}{\W_o} = 0$
which is tangent to the submanifold $\W_o \hookrightarrow \W$, and it is given locally by
$$
X = q_1 \derpar{}{q_0} + q_2 \derpar{}{q_1} + q_3 \derpar{}{q_2} - 
\frac{1}{\gamma}\left(\omega^2q_0 + q_2\right) \derpar{}{q_3} - \omega^2q_0\derpar{}{p^0} +
 \left(q_1 - p^0\right) \derpar{}{p_1} \ .
$$
Then, if $\sigma \colon \R \to \W$ is an integral curve of $X$ locally given by
\begin{equation}
\label{eqn:Example1_LocalCoordSigma}
\sigma(t) = \left(q_0(t),q_1(t),q_2(t),q_3(t),p^0(t),p^1(t)\right) \ ,
\end{equation}
and its component functions are solutions to the system
\begin{align}
& \dot{q}_0(t) = q_1(t); \label{eqn:Example1_IntCurveX_Semispray1_1} \\
& \dot{q}_1(t) = q_2(t); \label{eqn:Example1_IntCurveX_Semispray1_2} \\
& \dot{q}_2(t) = q_3(t); \label{eqn:Example1_IntCurveX_Semispray1_3} \\
& \dot{q}_3(t) = -\frac{1}{\gamma}\left(\omega^2q_0(t) + q_2(t)\right);
 \label{eqn:Example1_IntCurveX_EulerLagrange} \\
& \dot{p}^0(t) =  - \omega^2q_0(t); \label{eqn:Example1_IntCurveX_Hamiltonian1} \\
& \dot{p}^1(t) = q_1(t) - p^0(t). \label{eqn:Example1_IntCurveX_Hamiltonian2}
\end{align}

Finally we recover the Lagrangian and Hamiltonian solutions for this system.
For the Lagrangian solutions, as we
have shown in Lemma \ref{lemma:Cap06_LagVectField} and
Theorem \ref{thm:Cap06_CorrX-XL}, the Euler-Lagrange vector field is the unique
semispray of type $1$, $X_\Lag \in \vf(\Tan^3Q)$, such that
$X_\Lag \circ \pr_1 \circ j_o = \Tan\pr_1 \circ X \circ j_o$.
Thus this vector field $X_\Lag$ is locally given by
$$
X_\Lag = q_1 \derpar{}{q_0} + q_2 \derpar{}{q_1} + q_3 \derpar{}{q_2} - 
\frac{1}{\gamma}\left(\omega^2q_0 + q_2\right) \derpar{}{q_3} \ .
$$
For the integral curves of $X_\Lag$ we know from Proposition \ref{prop:Cap06_XLCI}
that if $\sigma \colon \R \to \W$ is an integral curve of $X$, then
$\sigma_\Lag = \pr_1 \circ \sigma$ is an integral curve of $X_\Lag$.
Thus, if $\sigma$ is given locally by (\ref{eqn:Example1_LocalCoordSigma}), then
$\sigma_\Lag$ has the following local expression
\begin{equation}\label{eqn:Example1_LocalCoordSigmaL}
\sigma_\Lag(t) = \left(q_0(t),q_1(t),q_2(t),q_3(t)\right) \ ,
\end{equation}
and its components satisfy equations (\ref{eqn:Example1_IntCurveX_Semispray1_1}),
(\ref{eqn:Example1_IntCurveX_Semispray1_2}), (\ref{eqn:Example1_IntCurveX_Semispray1_3})
and (\ref{eqn:Example1_IntCurveX_EulerLagrange}).
Notice that equations  (\ref{eqn:Example1_IntCurveX_Semispray1_1}),
(\ref{eqn:Example1_IntCurveX_Semispray1_2}) and
(\ref{eqn:Example1_IntCurveX_Semispray1_3})
state that $\sigma_\Lag$ is the canonical lifting of a curve in the basis, that is,
there exists a curve $\gamma \colon \R \to Q$ such that
$\tilde\gamma^3 = \sigma_\Lag$. Furthermore, equation
(\ref{eqn:Example1_IntCurveX_EulerLagrange}) is the Euler-Lagrange equation
for this system.

Now, for the Hamiltonian solutions, as $\Lag$ is a regular Lagrangian,
Theorem \ref{thm:Cap06_CorrX-Xh} states that there exists a unique
vector field $X_h = \Leg_*X_\Lag \in \vf(\Tan^*(\Tan Q))$ which is a solution to
the Hamilton equation. Hence, it is given locally by
\begin{equation*}
X_h = q_1 \derpar{}{q_0} + q_2 \derpar{}{q_1} - \omega^2q_0\derpar{}{p^0} + \left(q_1 - p^0\right) \derpar{}{p_1}
\end{equation*}
For the integral curves of $X_h$, Proposition \ref{prop:Cap06_XhCI} states that
if $\sigma_\Lag \colon \R \to \Tan^3Q$ is an integral curve of $X_\Lag$ coming
from an integral curve $\sigma$ of $X$, then $\sigma_h = \Leg \circ \sigma_\Lag$
is an integral curve of the vector field $X_h$. Therefore, if $\sigma$ is given
locally by (\ref{eqn:Example1_LocalCoordSigma}), then $\sigma_\Lag$ is given
by (\ref{eqn:Example1_LocalCoordSigmaL}) and so $\sigma_h$ can be locally
written
$$
\sigma_h(t) = \left(q_0(t),q_1(t),p^0(t),p^1(t)\right) \ ,
$$
and its components must satisfy equations (\ref{eqn:Example1_IntCurveX_Semispray1_1}),
(\ref{eqn:Example1_IntCurveX_Semispray1_2}), (\ref{eqn:Example1_IntCurveX_Hamiltonian1})
and (\ref{eqn:Example1_IntCurveX_Hamiltonian2}).
Notice that these equations are the standard Hamilton equations for this system.

\subsection{The second-order relativistic particle}

Let us consider a relativistic particle whose action is proportional to its extrinsic curvature. 
 This system was analyzed in \cite{art:Plyushchay88,art:Pisarski86,art:Batlle_Gomis_Pons_Roman88,art:Nesterenko89},
  and here we study it using the Lagrangian-Hamiltonian unified  formalism.
 
The configuration space is a $n$-dimensional smooth manifold $Q$
with local coordinates $(q_0^A)$, $1 \leqslant A \leqslant n$. 
Then, if we take the natural set of coordinates on the higher-order tangent bundles over $Q$,
the second-order Lagrangian function for this system,
$\Lag \in \Cinfty(\Tan^2Q)$, can be written locally as
\begin{equation}
\Lag(q_0^i,q_1^i,q_2^i) = \frac{\alpha}{(q_1^i)^2} \left[ (q_1^i)^2(q_2^i)^2 - (q_1^iq_2^i)^2 \right]^{1/2} 
\equiv \frac{\alpha}{(q_1^i)^2} \sqrt{g}  \ .
 \label{eqn:Example_Lagrangian}
\end{equation}
where $\alpha$ is some nonzero constant. It is a singular Lagrangian, as we
can see by computing the Hessian matrix of $\Lag$ with respect to $q_2^A$,
which is
$$
\left( \frac{\partial^2\Lag}{\partial q_2^B\partial q_2^A} \right) = \begin{cases}
\displaystyle \frac{\alpha}{2(q_1^i)^2\sqrt{g^3}} \left[ \left((q_1^iq_2^i)^2 - 2(q_1^i)^2(q_2^i)^2 \right)q_1^Bq_1^A \right. 
+ (q_1^i)^2(q_1^iq_2^i)(q_2^Bq_1^A-q_1^Bq_2^A)  & 
 \\
\displaystyle \qquad\qquad\quad \left. - (q_1^i)^2(q_2^i)^2q_2^Bq_2^A \right]
& \mbox{ if } B \neq A  
\\
\displaystyle \frac{\alpha}{\sqrt{g^3}}\left[ g - 
(q_2^i)^2q_1^Aq_1^A + 2(q_1^iq_2^i)q_1^Aq_2^A - (q_1^i)^2q_2^Aq_2^A \right] & \mbox{ if } B = A \ ,
\end{cases}
$$
then after a long calculation we obtain that
$\ds \det\left( \frac{\partial^2\Lag}{\partial q_2^B\partial q_2^A} \right) = 0$.
In particular, $\Lag$ is an almost-regular Lagrangian.

As this is a second-order dynamical system, the phase space that we consider is
$$
\xymatrix{
\ & \W = \Tan^3Q \times_{\Tan Q} \Tan^*(\Tan Q) \ar[dl]_-{\pr_1} \ar[dr]^-{\pr_2} & \ \\
\Tan^3Q \ar[dr]_-{\rho^{3}_{1}} & \ & \Tan^*(\Tan Q) \ar[dl]^-{\pi_{\Tan Q}} \\
\ & \Tan Q & \  .
}
$$
As $\Lag$ is almost-regular, the ``natural'' phase space for this system would be 
$\Tan^3Q \times_{\Tan Q} P_o$, where $P_o \hookrightarrow \Tan^*(\Tan Q)$
 denotes the image of the Legendre-Ostrogradsky map.
 However, as we have a set of natural coordinates defined in $\W$, 
 it is easier to work in $\W$ and then to obtain the constraints as a consequence of the formalism.

If $\omega_{1} \in \df^2(\Tan^*(\Tan Q))$ is the canonical symplectic form,
we define the presymplectic form
$\Omega = \pr_2^*\omega_{1}\in \df^2(\W)$,
whose local expression is
$$
\Omega = dq_0^i \wedge dp_i^0 + dq_1^i \wedge dp_i^1 \ .
$$
The Hamiltonian function $H \in \Cinfty(\W)$ is
$H = \C - (\rho^3_2 \circ \pr_1)^*\Lag$,
where $\C$ is the coupling function, whose local expression is
$\C\left(q_0^i,q_1^i,q_2^i,q_3^i,p_i^0,p_i^1\right) = p_i^0q_1^i + p_i^1q_2^i$,
and then the Hamiltonian function can be written locally
$$
H\left(q_0^i,q_1^i,q_2^i,q_3^i,p_i^0,p_i^1\right) =
 p_i^0q_1^i + p_i^1q_2^i - \frac{\alpha}{(q_1^i)^2} \left[ (q_1^i)^2(q_2^i)^2 - (q_1^iq_2^i)^2 \right]^{1/2} \ .
$$

The dynamics for this system are described as the integral curves of  vector fields
$X \in \vf(\W)$ which are solutions to equation (\ref{eqn:Cap06_EqDinImp}).
If we take a generic vector field $X\in\vf(\W)$, given locally by
$$
X = f_0^A \derpar{}{q_0^A} + f_1^A \derpar{}{q_1^A} + F_2^A \derpar{}{q_2^A} + 
F_3^A \derpar{}{q_3^A} + G_A^0 \derpar{}{p_A^0} + G_A^1 \derpar{}{p_A^1} \ ,
$$
taking into account that
\beann
\d H &=& \displaystyle q_1^A\d p_A^0 + q_2^A\d p_A^1 + 
\left[ p^0_A + \frac{\alpha}{((q_1^i)^2)^2\sqrt{g}}\left[ \left((q_1^i)^2(q_2^i)^2 - 2(q_1^iq_2^i)^2\right)q_1^A + 
(q_1^iq_2^i)(q_1^i)^2q_2^A \right] \right]\d q_1^A\\
& &+ \left[ p_A^1 - \frac{\alpha}{(q_1^i)^2\sqrt{g}}\left( (q^i_1)^2 q_2^A - (q_1^iq_2^i)q_1^A \right) \right]\d q_2^A \ ;
\eeann
from the dynamical equation we obtain the following linear systems for the coefficients of $X$
\begin{align}
 & f_0^A = q_1^A \label{eqn:Example_Semispray2_1} \\
 & f_1^A = q_2^A \label{eqn:Example_Semispray2_2} \\
& G_A^0 = 0 \\
 & G_A^1 = - p^0_A - 
\frac{\alpha}{((q_1^i)^2)^2\sqrt{g}}\left[ \left((q_1^i)^2(q_2^i)^2 - 2(q_1^iq_2^i)^2\right)q_1^A + 
(q_1^iq_2^i)(q_1^i)^2q_2^A \right] \label{eqn:Example_VectorFieldG1} \\
 & p_A^1 - \frac{\alpha}{(q_1^i)^2\sqrt{g}}\left( (q^i_1)^2 q_2^A - (q_1^iq_2^i)q_1^A \right) = 0 \ .
\label{eqn:Example_LegTransformation}
\end{align}
Note that from equations \eqref{eqn:Example_Semispray2_1} and \eqref{eqn:Example_Semispray2_2}
we obtain the condition of semispray of type $2$ for $X$.
Furthermore, equations \eqref{eqn:Example_LegTransformation} are algebraic relations
between the coordinates in $\W$ stating that the vector field $X$ is defined along a submanifold
$\W_o$ that is identified with the graph of the Legendre-Ostrogradsky map,
as we stated in Propositions \ref{prop:Cap06_ExistSolEqDin} and \ref{prop:Cap06_W0GrafFL}.
Thus, the vector field $X$ is given locally by
\begin{equation}
\label{eqn:Example_VectorFieldBeforeHolonomy}
X = q_1^A \derpar{}{q_0^A} + q_2^A \derpar{}{q_1^A} + F_2^A \derpar{}{q_2^A} +
F_3^A \derpar{}{q_3^A} + G_A^1 \derpar{}{p_A^1} \ ,
\end{equation}
where the functions $G_A^1$ are determined by \eqref{eqn:Example_VectorFieldG1}.
As we want to recover the Lagrangian solutions from the vector field $X$,
we must require $X$ to be a semispray of type $1$.
This condition reduces the set of vector fields $X \in \vf(\W)$ given by 
\eqref{eqn:Example_VectorFieldBeforeHolonomy} to the following ones
\begin{equation}
\label{eqn:Example_VectorFieldHolonomy}
X = q_1^A \derpar{}{q_0^A} + q_2^A \derpar{}{q_1^A} + q_3^A \derpar{}{q_2^A} + 
F_3^A \derpar{}{q_3^A} + G_A^1\derpar{}{p_A^1} \ .
\end{equation}
Notice that the functions $F_3^A$ are not determinated until the tangency of the vector field $X$
on $\W_o$ is required. Now, the Legendre-Ostrogradsky transformation is the map
$\Leg \colon \Tan^3Q \longrightarrow \Tan^*(\Tan Q)$ locally given by
\beann
\Leg^*(p^0_A) &=& 
\derpar{\Lag}{q_1^A} - d_T\left(\derpar{\Lag}{q_2^A}\right) \equiv \derpar{\Lag}{q^A_1} - d_T \left( p_A^1 \right) =
\\ & &
\frac{\alpha}{(q_1^i)^2\sqrt{g^3}} \left[ \left( (q_2^i)^2g + (q_1^i)^2(q_2^i)^2(q_1^iq_3^i) -
 (q_1^i)^2(q_1^iq_2^i)(q_ 2^iq_3^i) \right)q_1^A\right] + \\ & &
\frac{\alpha}{(q_1^i)^2\sqrt{g^3}} \left[ \left( ((q_1^i)^2)^2(q_2^iq_3^i) - (q_1^i)^2(q_1^iq_2^i)(q_1^iq_3^i) - 
(q_1^iq_2^i)g \right)q_2^A - (q_1^i)^2gq_3^A\right] \\ & &
\\ 
\Leg^*(p^1_A) &=& \derpar{\Lag}{q_2^A} =
\frac{\alpha}{(q^i_1)^2\sqrt{g}} \left[ (q^i_1)^2q_2^A - (q^i_1q^i_2)q^A_1 \right] \ ,
\eeann
and, in fact, $\Lag$ is an almost-regular Lagrangian.
Thus, from the expression in local coordinates of the map $\Leg$,
we obtain the (primary) constraints that define the closed submanifold $P_o={\rm Im}\,\Leg$,
which are
\beq
\phi^{(0)}_1 \equiv p^1_iq_1^i = 0 \quad ; \quad
\phi^{(0)}_2 \equiv (p_i^1)^2 - \frac{\alpha^2}{(q^i_1)^2} = 0 \ ;
\label{eqn:Example_Constraints0}
\eeq
Let $\Leg_o \colon \Tan^{3}Q \to P_o$.
Then, the submanifold $\W_o ={\rm graph}\,\Leg_o$ is defined by
$$
\W_o = \left\{ p \in \W \ \colon \ \xi^A_0(p) = \xi^A_1(p) = \phi^{(0)}_1(p) = \phi^{(0)}_2(p) = 0,
 \ 1 \leqslant A \leqslant \dim Q \right\}
$$
where $\xi_r^A \equiv p_A^r - \Leg^*p_A^r$. The diagram for this situation is
$$
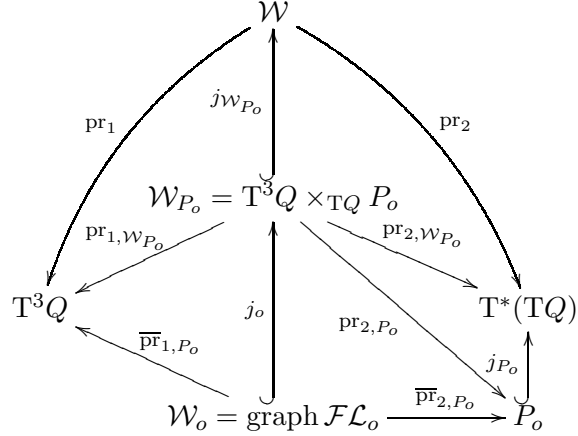
\xymatrix{
\ & \W \ar@/_1.25pc/[dddl]_-{\pr_1} \ar@/^1.25pc/[dddr]^-{\pr_2} & \ \\
\ & \ & \ \\
\ & \W_{P_o} = \Tan^{3}Q \times_{\Tan Q} P_o \ar@{^{(}->}[uu]^-{j_{\W_{P_o}}}
 \ar[dl]_-{\pr_{1,\W_{P_o}}} \ar[dr]^-{\pr_{2,\W_{P_o}}} \ar[ddr]_-{\pr_{2,P_o}}  & \ \\
\Tan^3Q & \ & \Tan^*(\Tan Q) \\
\ & \W_o = {\rm graph}\,\Leg_o \ar@{^{(}->}[uu]^-{j_o} \ar[ul]_-
{\overline{\pr}_{1,P_o}} \ar[r]^-{\overline{\pr}_{2,P_o}} & P_o \ar@{^{(}->}[u]^{j_{P_o}} 
}
$$
Notice that $\W_o$ is a submanifold of $\Tan^3Q \times_{\Tan Q} P_o$,
and that $\W_o$ is the real phase space of the system,
where the dynamics take place.

Next we compute the tangency condition for
$X \in \vf(\W)$ given locally by \eqref{eqn:Example_VectorFieldHolonomy} on the submanifold
 $\W_o \hookrightarrow \W_{P_o} \hookrightarrow \W$, by
checking if the following identities hold
\bea
\Lie(X)\xi^A_0\vert_{\W_o} = 0 & \ , \ &
\Lie(X)\xi^A_1\vert_{\W_o} = 0 
\label{eqn:Example_LagEquations1}\\
\Lie(X)\phi^{(0)}_1\vert_{\W_o} = 0  & \ , \ &
\Lie(X)\phi^{(0)}_2\vert_{\W_o} = 0 \ .
\label{eqn:Example_LieDerConstraints0}
\eea
As we have seen in Section \ref{section:DynamicsW}, equations (\ref{eqn:Example_LagEquations1})
give us the Lagrangian equations for the vector field $X$.
However, equations (\ref{eqn:Example_LieDerConstraints0}) do not hold since
$$
\Lie(X)\phi^{(0)}_1 = \Lie(X)(p^1_iq_1^i) = - p_i^0q_1^i \quad , \quad
\Lie(X)\phi^{(0)}_2 = \Lie(X)((p^1_i)^2 - \alpha^2 / (q_1^i)^2) = - 2p_i^0q_i^1 \ ,
$$
and hence we obtain two first-generation secondary constraints
\beq
\phi^{(1)}_1 \equiv p_i^0 q_1^i = 0 \quad , \quad
\phi^{(1)}_2 \equiv p_i^0p_i^1 = 0 
\label{eqn:Example_Constraints1}
\eeq
that define a new submanifold $\W_1 \hookrightarrow \W_o$.
Now, checking the tangency of the vector field $X$ to this new submanifold,
we obtain
$$
\Lie(X)\phi^{(1)}_1 = \Lie(X)(p_i^0q_1^i) = 0 \quad , \quad
\Lie(X)\phi^{(1)}_2 = \Lie(X)(p_i^0p_i^1) = -(p^0_i)^2 \ ,
$$
and a second-generation secondary constraint appears
\beq
\phi^{(2)} \equiv (p^0_i)^2 = 0 \ ,
\label{eqn:Example_Constraints2}
\eeq
which defines a new submanifold $\W_2 \hookrightarrow \W_1$.
Finally, the tangency of the vector field $X$ on this submanifold
gives no new constraints, since
$$
\Lie(X)\phi^{(2)} = \Lie(X)((p^0_i)^2) = 0 \ .
$$
So we have two primary constraints \eqref{eqn:Example_Constraints0},
 two first-generation secondary constraints \eqref{eqn:Example_Constraints1},
 and a single second-generation secondary constraint \eqref{eqn:Example_Constraints2}.
 Notice that these five constraints only depend on $q_1^A$, $p^0_A$ and $p^1_A$, 
 and so they are $\pr_2$-projectable. Thus, we have the following diagram
$$
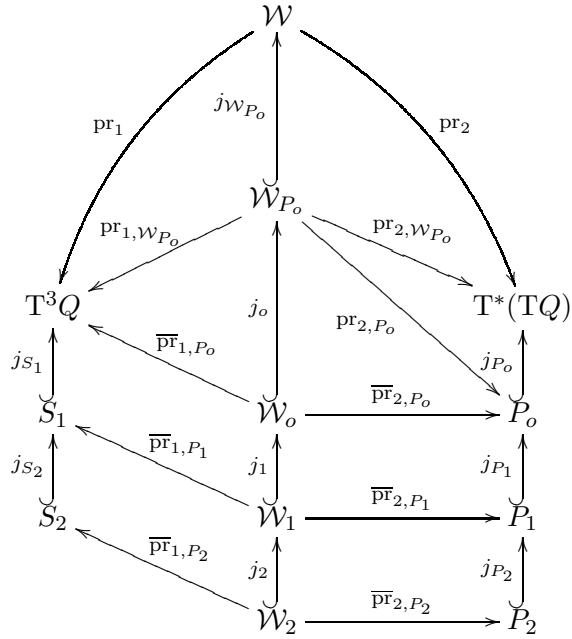
\xymatrix{
\ & \ & \W \ar@/_1.25pc/[dddll]_-{\pr_1} \ar@/^1.25pc/[dddrr]^-{\pr_2} & \ & \ \\
\ & \ & \ & \ & \ \\
\ & \ & \W_{P_o} \ar@{^{(}->}[uu]^-{j_{\W_{P_o}}}
 \ar[dll]_-{\pr_{1,\W_{P_o}}} \ar[drr]^-{\pr_{2,\W_{P_o}}} \ar[ddrr]_-{\pr_{2,P_o}} & \  & \ \\
\Tan^3Q & \ & \ & \ & \Tan^*(\Tan Q) \\
S_1 \ar@{^{(}->}[u]^-{j_{S_1}} & \ & \W_o \ar@{^{(}->}[uu]^-{j_o} \ar[ull]_-{\overline{\pr}_{1,P_o}} 
\ar[rr]^-{\overline{\pr}_{2,P_o}} & \ & P_o \ar@{^{(}->}[u]^{j_{P_o}} \\
S_2 \ar@{^{(}->}[u]^-{j_{S_2}} & \ & \W_1 \ar@{^{(}->}[u]^-{j_1} \ar[ull]_-{\overline{\pr}_{1,P_1}}
 \ar[rr]^-{\overline{\pr}_{2,P_1}} & \ & P_1 \ar@{^{(}->}[u]^-{j_{P_1}} \\
\ & \ & \W_2 \ar@{^{(}->}[u]^-{j_2} \ar[ull]_-{\overline{\pr}_{1,P_2}}
 \ar[rr]^-{\overline{\pr}_{2,P_2}} & \ & P_2 \ar@{^{(}->}[u]^-{j_{P_2}}
}
$$
where
\begin{align*}
&P_1 = \left\{ p \in P_o \colon \phi^{(1)}_1(p) = \phi^{(1)}_2(p) = 0 \right\} =\pr_2(\W_1) \\
&P_2 = \left\{ p \in P_o \colon \phi^{(2)}(p) = 0 \right\} =\pr_2(\W_2) \\
&S_1 = \Leg_o^{-1}(P_1) = \pr_1(\W_1)\\
&S_2 = \Leg_o^{-1}(P_2) = \pr_1(\W_2) \ .
\end{align*}
Focusing only on the Legendre-Ostrogradsky map, 
and ignoring the unified part of the diagram, we have
$$
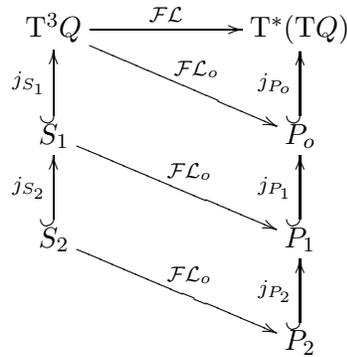
\xymatrix{
\Tan^3 Q \ar[rr]^-{\Leg} \ar[drr]^-{\Leg_o} & \ & \Tan^*(\Tan Q) \\
S_1 \ar@{^{(}->}[u]^-{j_{S_1}} \ar[drr]^-{\Leg_o} & \ & P_o \ar@{^{(}->}[u]^-{j_{P_o}} \\
S_2 \ar@{^{(}->}[u]^-{j_{S_2}} \ar[drr]^-{\Leg_o} & \ & P_1 \ar@{^{(}->}[u]^-{j_{P_1}} \\
\ & \ & P_2 \ar@{^{(}->}[u]^-{j_{P_2}}
}
$$

Notice that we still have to check \eqref{eqn:Example_LagEquations1}.
As we have seen in Section \ref{section:DynamicsW}, we will obtain the following equations
\begin{align}
&\left(F_3^B - d_T\left(q_3^B\right)\right)\frac{\partial^2\Lag}{\partial q_2^B\partial q_2^A} +
 \derpar{\Lag}{q_0^A} - d_T\left(\derpar{\Lag}{q_1^A}\right) + 
 d_T^2\left(\derpar{\Lag}{q_2^A}\right) + 
 \left(F_2^B - q_3^B\right)d_T\left(\frac{\partial^2\Lag}{\partial q_2^B\partial q_2^A}\right) = 0
  \label{eqn:Example_EulerLagrangeInitial} \\
&\left(F_2^B - q_3^B\right)\frac{\partial^2\Lag}{\partial q_2^B\partial q_2^A} = 0 
\label{eqn:Example_Semispray1}
\end{align}
As we have already required the vector field $X$ to be a semispray of type $1$,
equations \eqref{eqn:Example_Semispray1} are satisfied identically
and equations \eqref{eqn:Example_EulerLagrangeInitial} become
\begin{equation}\label{eqn:Example_EulerLagrangeFinal}
\left(F_3^B - d_T\left(q_3^B\right)\right)\frac{\partial^2\Lag}{\partial q_2^B\partial q_2^A} +
 \derpar{\Lag}{q_0^A} - 
d_T\left(\derpar{\Lag}{q_1^A}\right) + d_T^2\left(\derpar{\Lag}{q_2^A}\right) = 0 \ .
\end{equation}
A long calculation shows that this equation is compatible and so no new constraints arise.
Thus, we have no Lagrangian constraint appearing from the semispray condition.
 If some constraint had appeared, it would not be $\Leg_o$-projectable
 (see \cite{art:Gracia_Pons_Roman92})

Thus, the vector fields $X \in \vf(\W)$ given locally by \eqref{eqn:Example_VectorFieldHolonomy}
 which are solutions to the equation
$$
\restric{\left[\inn(X)\Omega - \d H\right]}{\W_2} = 0 \ ,
$$
are  tangent to the submanifold $\W_2 \hookrightarrow \W_o$.
Therefore, taking the vector fields $X_o\in\vf(\W_2)$ such that
$\Tan j_2\circ X_o=X\circ j_2$, the form
$\Omega_o = (j_{\W_{P_o}} \circ j_o \circ j_1 \circ j_2)^*\Omega$,
and the canonical Hamiltonian function $H_o = (j_{\W_{P_o}} \circ j_o \circ j_1 \circ j_2)^*H$,
the above equation leads to
$$
\label{eqn:Example_DynamicalEquationsRestricted}
\inn(X_o)\Omega_o - \d H_o= 0 \ ,
$$
but a simple calculation in local coordinates shows that $H_o=0$, and thus
the last equation becomes
$\inn(X_o)\Omega_o= 0$.

One can easily check that, if the semispray condition is not required at the beginning
and we perform all this procedure with the vector field given by \eqref{eqn:Example_VectorFieldBeforeHolonomy},
the final result is the same. This means that, in this case,
 the semispray condition does not give any additional constraint.

As final results, we recover the Lagrangian and Hamiltonian vector fields from the vector field $X \in \vf(\W)$.
For the Lagrangian vector field, by using Lemma \ref{lemma:Cap06_LagVectField} and
Theorem \ref{thm:Cap06_CorrX-XL} we obtain a semispray of type $2$,
$X_\Lag \in \vf(\Tan^3Q)$, tangent to $S_2$. Thus, requiring the condition of semispray of
type $1$ to be satisfied (perhaps on another submanifold $M_2 \hookrightarrow S_2$), the local
expression for the vector field $X_\Lag$ is
\begin{equation*}
X_\Lag = q_1^A \derpar{}{q_0^A} + q_2^A \derpar{}{q_1^A} + q_3^A \derpar{}{q_2^A} + 
F_3^A \derpar{}{q_3^A} \ .
\end{equation*}
where the functions $F_3^A$ are determined by (\ref{eqn:Example_EulerLagrangeFinal}).
For the Hamiltonian vector fields, recall that $\Lag$ is an almost-regular Lagrangian
function. Thus, we know that there are Euler-Lagrange vector fields which are
$\Leg_o$-projectable on $P_2$, tangent to $P_2$ and solutions to the Hamilton
equation.

\section{Conclusions and outlook}
\label{section:outlook}

After introducing the natural geometric structures needed for describing
 higher-order autonomous dynamical systems, we review their Lagrangian and Hamiltonian formalisms,
 following the exposition made in \cite{book:DeLeon_Rodrigues85}.
 
The main contribution of this work is that
we develop the Lagrangian-Hamiltonian unified formalism for higher-order dynamical systems,
following the ideas of the original article \cite{art:Skinner_Rusk83}.
We pay special attention to showing how the Lagrangian and Hamiltonian dynamics are recovered from 
this, both for regular and singular systems.

A first consideration is to discuss the
fundamental differences between the first-order and the higher-order
unified Lagrangian-Hamiltonian formalisms. In particular:
\begin{itemize}
\item
As there is no canonical pairing between the elements of $\Tan^{2k-1}_qQ$ and of $\Tan^*_q(\Tan^{k-1}Q)$,
in order to define the higher-order coupling function $\C$ in an intrinsic way,
we use the canonical injection that transforms a point in $\Tan^{2k-1}Q$
into a tangent vector along $\Tan^{k-1}Q$.
\item
When the equations that define the Legendre-Ostrogradsky map
are recovered from the unified formalism
(both in the characterization of the compatibility submanifold $\W_o$
as the graph of $\Leg$, and in the equations in local coordinates of the vector field
$X \in \vf(\W)$ solution to the dynamical equations),
the only equations that are recovered are those that define the highest order momentum coordinates,
and the remaining equations that define the map must be recovered
using the relations between the momentum coordinates.
\item
The regularity of the Lagrangian function is more relevant in the higher-order case,
because the condition of semispray of type $1$ (the holonomy condition) of the Lagrangian vector field
cannot be deduced from the dynamical equations if the Lagrangian is singular,
unlike the first-order case, where this holonomy condition is deduced straightforwardly
from the equations independently of the regularity of the Lagrangian function.
When the Lagrangian is singular, we can only ensure that the Lagrangian vector field is a semispray of type $k$.
It is therefore necessary, in general, to require the condition of semispray of type $1$
as an additional condition.

Then, for regular Lagrangian systems, when the tangency condition of the vector field
$X \in \vf(\W)$ solution in the unified formalism along the submanifold $\W_o$ is required,
we obtain not only the Euler-Lagrange equations for the vector field,
but also the remaining $k-1$ systems of equations that the vector field must satisfy to be a semispray of type $1$.
\end{itemize}

As we point out in the introduction,
a previous and quick presentation of a unified formalism for higher-order systems
was outlined in \cite{art:Colombo_Martin_Zuccalli10}.
Our formalism differs from this one, since in that article the authors
take $\Tan^kQ \oplus_{\Tan^{k-1}Q} \Tan^*(\Tan^{k-1}Q)$ as
the phase space in the unified formalism, instead of ours, which is
$\Tan^{2k-1}Q \oplus_{\Tan^{k-1}Q} \Tan^*(\Tan^{k-1}Q)$.
This is a significant difference, since when we want to recover the dynamical
solutions of the Lagrangian formalism from the unified formalism,
the Lagrangian phase space is $\Tan^{2k-1}Q$, instead of $\Tan^kQ$,
which is the bundle where the Lagrangian function is defined.
This fact makes it more natural to obtain the Lagrangian dynamics as well as the
Hamiltonian dynamics, which in turn is obtained from the Lagrangian one using the Legendre-Ostrogradsky map.

By using any suitable generalization of some of the several formalisms for first-order non-autonomous dynamical systems
\cite{book:Abraham_Marsden78,art:Barbero_Echeverria_Martin_Munoz_Roman08,art:Echeverria_Munoz_Roman91},
a future avenue of research consists in generalizing this unified formalism for higher-order
non-autonomous dynamical systems.
This generalization should also be recovered as a particular case of the corresponding unified formalism
for higher-order classical field theories.
As regards this topic, a proposal for a unified formalism for  higher-order classical field theories has recently been made
\cite{art:Campos_DeLeon_Martin_Vankerschaver09,art:Vitagliano10}, which is based on the model presented in \cite{art:Colombo_Martin_Zuccalli10}.
This formulation allows us to improve some
previous models for
describing the Lagrangian and Hamiltonian formalisms.
Nevertheless, some ambiguities arise when considering the solutions of the field equations.
We hope that a suitable extension of our formalism to field theories will enable
these difficulties to be overcome and complete the model given in \cite{art:Campos_DeLeon_Martin_Vankerschaver09,art:Vitagliano10}.

\section*{Acknowledgments}

We acknowledge the financial support of the  {\sl
Ministerio de Ciencia e Innovaci\'on} (Spain), projects
MTM2008-00689 and MTM2009-08166-E.
We also  thank Mr. Jeff Palmer for his assistance in preparing the English
version of the manuscript.

{\small
\bibliography{Bibliografia}
\bibliographystyle{AMS_Mod}
}
\end{document}